\documentclass{article}
\usepackage[utf8]{inputenc}
\usepackage{authblk}
\usepackage{setspace}
\usepackage[margin=1.25in]{geometry}
\usepackage{graphicx}
\graphicspath{ {} }
\usepackage{subcaption}
\usepackage{amsmath}
\usepackage{amssymb}
\usepackage{mathtools}
\usepackage{amsthm}
\usepackage{lineno}
\usepackage{color}

\theoremstyle{plain}
\newtheorem{theorem}{Theorem}[section]

\newtheorem{corollary}[theorem]{Corollary}
\theoremstyle{definition}

\theoremstyle{remark}
\newtheorem{remark}{Remark}

\title{ Tracking control of latent dynamic systems with application to spacecraft attitude control}
\author[1]{Congxi Zhang}
\author[1*]{Yongchun Xie} 

\affil[1]{Beijing Institute of Control Engineering, Beijing, P.~R.~China.} 
\affil[*]{Address correspondence to: xieyongchun@vip.sina.com} 

\date{}

\begin{document}

\maketitle

\begin{abstract}
 When intelligent spacecraft or space robots perform tasks in a complex environment, the controllable variables are usually not directly available and have to be inferred from high-dimensional observable variables, such as outputs of neural networks or images. While the dynamics of these observations are highly complex, the mechanisms behind them may be simple, which makes it possible to regard them as latent dynamic systems. For control of latent dynamic systems, methods based on reinforcement learning suffer from sample inefficiency and generalization problems. In this work, we propose an asymptotic tracking controller for latent dynamic systems. The latent variables are related to the high-dimensional observations through an unknown nonlinear function. The dynamics are unknown but assumed to be affine nonlinear. To realize asymptotic tracking, an identifiable latent dynamic model is learned to recover the latents and estimate the dynamics. This training process does not depend on the goals or reference trajectories. Based on the learned model, we use a manually designed feedback linearization controller to ensure the asymptotic tracking property of the closed-loop system. After considering fully controllable systems, the results are extended to the case that uncontrollable environmental latents exist. {As an application, simulation experiments on a latent spacecraft attitude dynamic model are conducted to verify the proposed methods, and the observation noise and control deviation are taken into consideration.}
\end{abstract}


\section{Introduction}
 Designing controllers to ensure the stability or tracking property of dynamic systems is a long-standing problem. { Take spacecraft attitude control as an example. Classic control methods include proportional integral derivative (PID) control, sliding mode control \cite{tiwari2015rigid,miao2019adaptive}, model predictive control \cite{hegrenaes2005spacecraft, dong2020tube}, characteristic model-based control \cite{wu2007characteristic}, and so on. Recently, intelligent control methods, such as fuzzy control \cite{chen2022intelligent}, neural network compensation \cite{sun2023satellite}, and reinforcement learning based gain tuning\cite{wei2023adaptive}, have been widely used in the identification and control problems of spacecraft attitude control systems.} Most of these methods, either manually designed controllers \cite{tiwari2015rigid,miao2019adaptive,hegrenaes2005spacecraft, dong2020tube,wu2007characteristic,khalil2002control,slotine1991applied} or data-driven approaches \cite{chen2022intelligent,sun2023satellite,wei2023adaptive,hunt1992neural,brunke2022safe,lederer2019uniform,chang2019neural,sun2021learning,richards2023learning}, require the state variables are low-dimensional and the dynamic models are reasonable, i.e., the model is a representation of some invariant physical mechanisms. However, in complex environments, an intelligent spacecraft may encounter high-dimensional observations such as outputs of neural networks or images \cite{scholkopf2021toward,tomar2023learning,yuanli,lilinfeng}, and the dynamics of these observations are highly nonlinear and unknown \cite{bengio2013representation,scholkopf2021toward}. These kinds of systems are generally regarded as latent dynamic systems.

A standard method to realize control of latent dynamic systems is end-to-end reinforcement learning (RL) \cite{singh2019end,andrychowicz2020learning}. While RL has shown great success in recent years, they suffer from sample inefficiency \cite{botvinick2019reinforcement,wen2023large} and  generalization issues \cite{kirk2023survey}. {Another line of works first learns a latent representation or a latent dynamic model. Then they use RL to train the control policy based on the learned representation \cite{tomar2023learning,gelada2019deepmdp,wu2023daydreamer,laskin2020curl,schwarzer2020data,zhang2020learning}, or use model predictive control \cite{watter2015embed,lenz2015deepmpc,hafner2019dream}. \cite{hafner2019learning} uses the reconstruction loss to learn a recurrent state-space model for planning. \cite{zhang2020learning} claims that reconstruction-based representations involve task-irrelevant details and use bisimulation loss to abstract only the task-relevant information. \cite{fu2021learning} tries to disentangle the task-relevant and task-irrelevant features. \cite{brandfonbrener2024inverse} pretrains an inverse dynamics model to recover the states.} These methods are based on the assumption that given the environments are complex, the latent variables, which can represent the mechanisms behind these observations, are generally low-dimensional, and corresponding dynamic models are simple \cite{scholkopf2021toward,mohammadi2023neural,jaffe2024learning}. Leveraging this assumption improves the sample efficiency and interpretability. However, most works still lack stability guarantees and the ability to generalize to new goals or reference trajectories without online planning. Moreover, the learned representations may be entangled with environmental latents  that are irrelevant to control or decision-making. This makes it hard to generalize to unseen environments \cite{wang2022denoised,huang2022action,liu2024learning}.

When designing a provably stable and generalizable controller for latent dynamic systems, one major challenge is how to ensure that the learned latent representation is identifiable and that the learned model is a suitable estimate. This problem is possible to be solved by using identifiable representation learning approaches \cite{locatello2019challenging,khemakhem2020variational,lachapelle2022disentanglement,ahuja2021properties,ahuja2022weakly,yao2022temporally,lippe2022citris}, which try to determine the relationship between the real latent variables and the learned ones. {\cite{yao2022temporally,lippe2022citris} use independent component analysis to identify the latent temporal processes. \cite{liu2023learning} learns an identifiable world model to disentangle the latents relevant/irrelevant to the action or the reward. \cite{zhang2024} recovers the identifiable representation and model for a class of affine nonlinear discrete-time latent dynamic systems.} However, to our knowledge, there is no identifiable representation learning result for continuous-time latent dynamic systems which provably infer the controllable latents from high-dimensional observations, { and none of these works consider the control problem.}

In this work, we show that it is possible to realize asymptotic tracking control for latent dynamic systems using identifiable representation and model learning methods and manually designed controllers. To be specific, we propose an asymptotic tracking controller for a class of latent dynamic systems, where the latent variables are related to the high-dimensional observations through an unknown nonlinear injective function, and the dynamics are affine nonlinear with some additional assumptions. 

We first train an identifiable latent dynamic model to recover the latents and estimate the dynamics. This learning approach is inspired by the mechanism-based perspective \cite{ahuja2021properties,ahuja2022weakly} and the results in \cite{zhang2024}, while we study continuous-time dynamic systems in this work. Then, we design a feedback linearization controller based on the learned latent dynamic model to ensure the tracking property. We consider both fully controllable systems and the case that uncontrollable environmental latents exist and show the results hold for unseen environmental dynamics. The main contributions are summarized as follows.

\begin{itemize}
	\item We propose an identifiable representation and model learning approach for a class of continuous-time affine nonlinear latent dynamic systems. Based on the learned latent dynamic model, we design a feedback linearization controller to ensure the asymptotic tracking property of the closed-loop systems.  
	\item We extend our results to the case where uncontrollable environmental latent variables exist and show that if the latent dynamic model is trained in several different training environments, the identifiability results and asymptotic tracking property hold for unseen test environment dynamics.
	\item {As an application, we conduct simulation experiments on a latent spacecraft attitude dynamic model and analysis the effect of the observation noise and control deviation.}
\end{itemize}

\section{Preliminaries}

\subsection{Latent Dynamic Systems}
A standard latent dynamic system consists of a dynamic model and a mixing function. The dynamic model describes the mechanism behind the observations, and the mixing function determines the relationship between latent variables and observable variables.

In this work, we consider dynamic systems which can be modeled as 
\begin{equation}
	\boldsymbol{q}^{(k)}=F(\boldsymbol{q}, \boldsymbol{\dot q},\cdots,\boldsymbol{q}^{(k-1)}) + \boldsymbol{B}(\boldsymbol{q}, \boldsymbol{\dot q},\cdots,\boldsymbol{q}^{(k-1)}) \boldsymbol{u},
	\label{dynamic model}
\end{equation}
where $\boldsymbol{q} \in \mathbb{R}^n$ is a time-varying variable, and $\forall i=1,\cdots,k$, $\boldsymbol{q}^{(i)}\doteq \frac{\mathrm{d}^i \boldsymbol{q}}{\mathrm{d} t^i}$ denotes the $i$-th time derivative of $\boldsymbol{q}$. $\boldsymbol{u} \in \mathbb{R}^n$ is the control input, which is also a time-varying variable. For simplicity of notation, we drop the dependency on time $t$ from all the variables in this paper when there is no confusion. Function $F: \mathbb{R}^{n k} \to \mathbb{R}^{n}$ is bounded on bounded sets. Functional matrix $\boldsymbol{B}: \mathbb{R}^{n k} \to \mathbb{R}^{n\times n}$ is diagonal and bounded on bounded sets. We let $\boldsymbol{z} \doteq [\boldsymbol{q}^T, \boldsymbol{\dot q}^T, \cdots, (\boldsymbol{q}^{(k-1)})^T]^T$ to denote the state variable in system (\ref{dynamic model}), and let $b_i(\boldsymbol{z})$ to denote the $i$-th component in the diagonal of $\boldsymbol{B}(\boldsymbol{z})$. Throughout this work, we assume $\forall b_i(z)\geq b_h>0$, where $b_h$ is an unknown positive constant. Both $F(\cdot)$ and $\boldsymbol{B}(\cdot)$ are unknown, but the order of (\ref{dynamic model}), i.e., $k$, is assumed to be known for simplicity. Otherwise,  $k$ should be determined during training.

 The state variable $\boldsymbol{z}$ is latent, i.e., it is not directly measured, both in the training and control processes,  but is related to an observation data $\boldsymbol{x} \in \mathcal{X} \subset \mathbb{R}^m$ as
\begin{equation}
	\boldsymbol{x} = g(\boldsymbol{z}),
	\label{g}
\end{equation}
where $g: \mathbb{R}^{nk} \to \mathcal{X}$ is an unknown mixing function, e.g., a neural network. $g$ is assumed to be injective (An injective function maps distinct elements of its domain to distinct elements) and $k$ times differentiable. The observation $\boldsymbol{x}$ and the input $\boldsymbol{u}$ are accessible for training and control.

The dynamic model (\ref{dynamic model}) and the mixing function (\ref{g}) together constitute the latent dynamic system we want to identify and control in this work.

\subsection{Asymptotic Tracking of Latent Dynamic Systems}

Our objective is to propose a controller to let the state variable $\boldsymbol{z}$ asymptotically track a bounded desired latent trajectory $\boldsymbol{z}_d\doteq [\boldsymbol{q}_d^T, \boldsymbol{\dot q}_d^T, \cdots, (\boldsymbol{q}_d^{(k-1)})^T]^T$, i.e., if we define the tracking error $\boldsymbol{e}\doteq \boldsymbol{z}-\boldsymbol{z}_d$, then as $t \to \infty$, $\boldsymbol{e} \to 0$.

The latent trajectory $\boldsymbol{z}_d$ is not directly given. Instead, we have the corresponding observation space trajectory $\boldsymbol{x}_d=g(\boldsymbol{z}_d)$.

\begin{remark}
When referring to stability in data-driven control methods, there are mainly two kinds of works. The first kind studies the stability during the training process and ensures the convergence of the learnable modules, e.g., neural networks \cite{nguyen2022analytic,nguyen2022layer}. Another line of work does not care about the training process. Instead, they directly use the approximation ability of multilayer neural networks \cite{hornik1989multilayer} to represent the dynamic models, controllers, or Lyapunov functions and ensure the stability of closed-loop systems \cite{chang2019neural,lederer2019uniform,sun2021learning}. Our work belongs to the second type. What makes this work different is that we study the control problem for latent dynamic systems, where the state is not accessible, and both the (highly nonlinear) representation function and the corresponding dynamic model need to be learned.
\end{remark}

 \section{Materials and Methods}
 \subsection{Learning and Control for Fully Controllable Systems}
 In this section, we show how to identify and control a fully controllable latent dynamic system. The control problem is solved by first learning an identifiable latent dynamic model (Section \ref{Sec Representation and Model Learning}) and then designing a feedback linearization controller based on it (Section \ref{Sec Controller Design}). The extensions to systems with uncontrollable environmental latents will be discussed in Section \ref{Sec Extensions to Systems with Uncontrollable Latents}. 
 \subsubsection{Representation and Model Learning}
 \label{Sec Representation and Model Learning}
 
 In this section, we develop the learning method to recover the latents and estimate the model, and provide the identifiability result, which ensures the relationship between the learned representations and the actual latent variables. 

Notice that the observation $\boldsymbol{x}$ is attainable but high-dimensional, and the dynamics of $\boldsymbol{x}$ are usually highly nonlinear. Direct learning the dynamics of $\boldsymbol{x}$ is laborious or even impossible. Hence, the controllable latents, which are sufficient and necessary for designing a controller in this case, have to be inferred from $\boldsymbol{x}$. We use a representation function $h:\mathcal{X} \to \mathbb{R}^{nk}$ to infer the latents, which is defined as follows.
 \begin{equation}
 	\boldsymbol{\hat z} = h(\boldsymbol{x}),
 	\label{h}
 \end{equation}
 where $\boldsymbol {\hat z} = [\boldsymbol {\hat q}_0^T, \boldsymbol {\hat q}_1^T, \cdots, \boldsymbol {\hat q}_{k-1}^T]^T$ is the estimate of $\boldsymbol {z}$, and $\forall i=1,\cdots,k-1$, $\boldsymbol {\hat q}_i$ is the estimate of $\boldsymbol{q}^{(i)}$ (Recall that $\boldsymbol{q}^{(i)}$ is the $i$-th time derivative of $\boldsymbol{q}$). $h(\cdot)$ is assumed to be $k$ times differentiable, and in practice, it is usually approximated by a learnable neural network.
 
 The estimated dynamic model is governed by the estimated vector field 
 \begin{equation}
 	 \hat {\mathcal{F}}(\boldsymbol{\hat z}, \boldsymbol{u})\doteq\begin{bmatrix}\boldsymbol{\hat q}_1\\ \vdots \\ \boldsymbol{\hat q}_{k-1} \\ \hat F(\boldsymbol{\hat z})+  \boldsymbol{\hat B}(\boldsymbol{\boldsymbol{\hat z}})\boldsymbol{u} \end{bmatrix},
 	\label{estimated dynamic model}
 \end{equation}
 where $\hat F: \mathbb{R}^{n k} \to \mathbb{R}^{n}$ is the estimate of $F(\cdot)$. $\boldsymbol{\hat B}: \mathbb{R}^{n k} \to \mathbb{R}^{n\times n}$ is the estimate of $\boldsymbol{B}(\cdot)$, which is a diagonal function matrix and all the components in the diagonal $\hat b(\cdot) \geq \hat b_h>0$, where $\hat b_h$ is a positive constant. In practice, $\hat F(\cdot)$ and $\boldsymbol{\hat B}(\cdot)$ are also modeled by learnable neural networks to make sure they can approximate any function.

 The representation function (\ref{h}) and the estimated dynamic model (\ref{estimated dynamic model}) are learned by solving 
 \begin{align}
 	\mathop{\min}_{\substack{\hat F(\cdot), \boldsymbol{\hat B}(\cdot), h(\cdot)}}&  \mathop{E}   [|| \boldsymbol{\dot {\hat z}}- \hat {\mathcal{F}}(\boldsymbol{\hat z}, \boldsymbol{u}) ||^2].\notag \\
 	\text{s.t.}\quad \ \ \ &  |\hat b(\cdot)|\geq \hat h_b>0.
 	\label{minimizing}
 \end{align}

 The expectation is taken over all the states $\boldsymbol{x}\in \mathcal{X}$, generated by all $\boldsymbol{z} \in  \mathbb{R}^{nk}$ through $g(\cdot)$, and all $\boldsymbol{u} \in \mathbb{R}^{n}$. The distributions are not necessarily known. The constraint $|\hat b(\cdot)|\geq \hat h_b>0$ can be satisfied simply by using a positive activation function
 in the output layer of $\boldsymbol{\hat B}(\cdot)$. Since in practice, $\boldsymbol{\dot {\hat z}}$ can not be directly acquired, we use the forward difference method to approximate it by following \cite{brunton2016discovering,fries2022lasdi,bonneville2024gplasdi,park2024tlasdi}, i.e.,
 \begin{equation}
 \boldsymbol{\dot {\hat z}} \approx \frac{\boldsymbol{\hat z}(t+\Delta t)-\boldsymbol{\hat z}(t)}{\Delta t}.
 	\label{forward difference}
 \end{equation}
 
 The following theorem guarantees the property of the learned latent dynamic models.
 \begin{theorem}
 	Suppose the latent dynamic system is given by (\ref{dynamic model}) and (\ref{g}). If we learn a model (\ref{h}) and (\ref{estimated dynamic model}) by solving (\ref{minimizing}), then $\boldsymbol{\dot {\hat z}}= \hat {\mathcal{F}}(\boldsymbol{\hat z}, \boldsymbol{u})$, $\boldsymbol{q}$ is identified up to an invertible componentwise transformation, and $\tau\doteq h\circ g$ is a diffeomorphism.
 	\label{theorem fully controlled}
 \end{theorem}	
 
 When saying that  $\boldsymbol{q}$ is identified up to an invertible componentwise transformation, it means there exists an invertible componentwise transformation, denoted as $\tau_q(\cdot)$, such that $\boldsymbol{\hat q}_0=\tau_q(\boldsymbol{q})$ \cite{ahuja2022weakly,yao2022temporally,lippe2022citris}. The prove of Theorem \ref{theorem fully controlled} is given as follows.
 
 \begin{proof}
 	By following  \cite{ahuja2022weakly}, solving (\ref{minimizing}) immediately  gives $\boldsymbol{\dot {\hat z}}= \hat {\mathcal{F}}(\boldsymbol{\hat z}, \boldsymbol{u})$, i.e.,
 	\begin{equation}
 		\boldsymbol{\dot {\hat z}}=\begin{bmatrix} \boldsymbol{\dot {\hat q}}_0\\ \vdots \\ \boldsymbol{\dot {\hat q}}_{k-2} \\ \boldsymbol{\dot {\hat q}}_{k-1} \end{bmatrix}=\begin{bmatrix} \boldsymbol{\hat q_1}\\ \vdots \\ \boldsymbol{\hat q}_{k-1} \\ \hat F(\boldsymbol{\hat z})+ \boldsymbol{\hat B}(\boldsymbol{\hat z})\boldsymbol{u} \end{bmatrix}.
 		\label{dot hat z}
 	\end{equation} 
 	This is because the expectation in (\ref{minimizing}) is lower bounded by 0, and there exist models (e.g., the actual model) that obtain this global minimum.  
 	
 	We then rewrite the dynamic system (\ref{dynamic model}) in its state space form as 
 	\begin{equation}
 		\boldsymbol{\dot z} = \begin{bmatrix} \boldsymbol{\dot q}\\ \vdots \\  \boldsymbol{q}^{(k-1)} \\ F(\boldsymbol{z})+ \boldsymbol{B}(\boldsymbol{z})\boldsymbol{u} \end{bmatrix}.
 		\label{dot z}
 	\end{equation}

 	Since both $h$ and $g$ are $k$ times differentiable functions, $\tau= h\circ g$ is also $k$ times differentiable. Taking the time derivative of both sides in $\boldsymbol{\hat z} = \tau(\boldsymbol{z})$ gives
 	\begin{equation}
 		\boldsymbol{\dot {\hat z}} = \boldsymbol{J}(\boldsymbol{z}) \boldsymbol{\dot z},
 		\label{J_z}
 	\end{equation}
 	where $\boldsymbol{J}(\boldsymbol{z})\in \mathbb{R}^{nk \times nk}$ is the the Jacobian of $\tau(\boldsymbol{z})$. Then from (\ref{dot hat z}), (\ref{dot z}), and (\ref{J_z}), we obtain
 	\begin{equation}
 		\begin{bmatrix} \boldsymbol{\hat q}_1\\ \vdots \\ \boldsymbol{\hat q}_{k-1} \\ \hat F(\boldsymbol{\hat z})+ \boldsymbol{\hat B}(\boldsymbol{\hat z})\boldsymbol{u} \end{bmatrix} = \boldsymbol{J}(\boldsymbol{z}) \begin{bmatrix} \boldsymbol{\dot q}\\ \vdots \\  \boldsymbol{q}^{(k-1)} \\   F(\boldsymbol{z})+   \boldsymbol{B}(\boldsymbol{z})\boldsymbol{u} \end{bmatrix}.
 		\label{J_z ex}
 	\end{equation}
 	Taking the derivative of above equation w.r.t. $\boldsymbol{u}$ gives
 	\begin{equation}
 		\begin{bmatrix} \boldsymbol{0}\\ \vdots \\ \boldsymbol{0} \\ \boldsymbol{\hat B}(\boldsymbol{\hat z}) \end{bmatrix} = \boldsymbol{J}(\boldsymbol{z}) \begin{bmatrix} \boldsymbol{0}\\ \vdots \\ \boldsymbol{0} \\ \boldsymbol{B}(\boldsymbol{z}) \end{bmatrix}.
 	\end{equation}
 	
 	Let $\boldsymbol{J}_{i,j}(\boldsymbol{z}) \in \mathbb{R}^{n \times n}$ denote the $(i,j)$-th $n \times n$ block of $\boldsymbol{J}(\boldsymbol{z})$, then above equation gives that $\forall i=1,\cdots,k-1$, $\boldsymbol{J}_{i,k}(\boldsymbol{z})=\boldsymbol{0}$ and $\boldsymbol{J}_{k,k}(\boldsymbol{z})=\boldsymbol{\hat B}(\boldsymbol{\hat z})\boldsymbol{B}^{-1}(\boldsymbol{z})$. This means that $\forall i=0,\cdots,k-2$, $\boldsymbol{\hat q}_i$ is not a function of $\boldsymbol{q}^{k-1}$, and hence $\forall i,j=1,\cdots,k-1$, $\boldsymbol{J}_{i,j}(\boldsymbol{z})$ is  not a function of $\boldsymbol{q}^{k-1}$, which implies $\forall i,j=1,\cdots,k-1$, $\boldsymbol{J}_{i,j}(\boldsymbol{z})=\boldsymbol{J}_{i,j}(\boldsymbol{q},\boldsymbol{\dot q},\dots,\boldsymbol{q}^{(k-2)})$. Hence we can rewrite the first $n(k-1)$ row of equation (\ref{J_z ex}) as 
 	\begin{equation}
 		\begin{bmatrix} \boldsymbol{\hat q}_1\\ \vdots \\ \boldsymbol{\hat q}_{k-1}\end{bmatrix} = \boldsymbol{J}_{1:k-1}(\boldsymbol{q},\boldsymbol{\dot q},\dots,\boldsymbol{q}^{(k-2)}) \begin{bmatrix} \boldsymbol{\dot q}\\ \vdots \\  \boldsymbol{q}^{(k-1)} \end{bmatrix}, 
 		\label{J_1:k-1}
 	\end{equation}
 	where
 	\begin{equation}
 		  \boldsymbol{J}_{1:k-1}(\boldsymbol{q},\boldsymbol{\dot q},\dots,\boldsymbol{q}^{(k-2)}) \doteq \begin{bmatrix} \boldsymbol{J}_{11}(\boldsymbol{z})&\cdots& \boldsymbol{J}_{1,k-1}(\boldsymbol{z}) \\ \vdots & \ddots & \vdots\\  \boldsymbol{J}_{k-1,1}(\boldsymbol{z})&\cdots& \boldsymbol{J}_{k-1,k-1}(\boldsymbol{z}) \end{bmatrix}. 
 	\end{equation}

 	Taking the time derivative of equation (\ref{J_1:k-1}) gives
 	\begin{equation}
 		\begin{bmatrix} \boldsymbol{\hat q}_2\\ \vdots \\ \boldsymbol{\hat q}_{k-1}\\ \hat F(\boldsymbol{\hat z})+ \boldsymbol{\hat B}(\boldsymbol{\hat z})\boldsymbol{u} \end{bmatrix} =  \boldsymbol{\dot J}_{1:k-1}(\boldsymbol{q},\boldsymbol{\dot q},\dots,\boldsymbol{q}^{(k-2)}) \begin{bmatrix} \boldsymbol{\dot q}\\ \vdots \\  \boldsymbol{q}^{(k-1)}\end{bmatrix}  +   \boldsymbol{J}_{1:k-1}(\boldsymbol{q},\boldsymbol{\dot q},\dots,\boldsymbol{q}^{(k-2)}) \begin{bmatrix} \boldsymbol{\ddot q}\\ \vdots \\ \boldsymbol{q}^{(k-1)}\\ F(\boldsymbol{z})+   \boldsymbol{B}(\boldsymbol{z})\boldsymbol{u}  \end{bmatrix}
 	\end{equation}
 	Taking the derivative of above equation w.r.t. $\boldsymbol{u}$ gives
 	\begin{equation}
 		\begin{bmatrix} \boldsymbol{0}\\ \vdots \\ \boldsymbol{0} \\ \boldsymbol{\hat B}(\boldsymbol{\hat z})  \end{bmatrix} = \boldsymbol{J}_{1:k-1}(\boldsymbol{q},\boldsymbol{\dot q},\dots,\boldsymbol{q}^{(k-2)}) \begin{bmatrix} \boldsymbol{0}\\ \vdots \\ \boldsymbol{0} \\ \boldsymbol{B}(\boldsymbol{z}) \end{bmatrix}.
 	\end{equation}
 	Above equation gives that  $\forall i=1,\cdots,k-2$, $\boldsymbol{J}_{i,k-1}(\boldsymbol{z})=\boldsymbol{J}_{i,k-1}(\boldsymbol{q},\boldsymbol{\dot q},\dots,\boldsymbol{q}^{(k-2)})=\boldsymbol{0}$ and $\boldsymbol{J}_{k-1,k-1}(\boldsymbol{z})=\boldsymbol{\hat B}(\boldsymbol{\hat z})\boldsymbol{B}^{-1}(\boldsymbol{z})$, which means $\forall i=0,\cdots,k-3$, $\boldsymbol{\hat q}_i$ is not a function of $\boldsymbol{q}^{k-2}$ and $\forall i,j=1,\cdots,k-2$, $\boldsymbol{J}_{i,j}(\boldsymbol{z})$ is  not a function of $\boldsymbol{q}^{k-2}$.
 	
 	By recursively using above procedures we can conclude that $\boldsymbol{\hat q}_0$ is only a function of $\boldsymbol{q}$, i.e., $\boldsymbol{\hat q}_0=\tau_q(\boldsymbol{q})$, and $\boldsymbol{J}_{1,1}(\boldsymbol{z})=\boldsymbol{J}_{1,1}(\boldsymbol{q})=\boldsymbol{\hat B}(\boldsymbol{\hat z})\boldsymbol{B}^{-1}(\boldsymbol{z})$, where $\boldsymbol{J}_{1,1}(\cdot)$ is also the Jacobian of $\tau_q(\cdot)$. 
 	
 	Since $\boldsymbol{\hat B}(\boldsymbol{\hat z})$ and $\boldsymbol{B}(\boldsymbol{z})$ are all positive definite diagonal matrices, $\boldsymbol{J}_{1,1}(\boldsymbol{q})$ is also a positive definite diagonal matrix, which gives that $\tau_q(\cdot)$ is a componentwise invertible function.
 	
 	We also obtain that $\forall i=1 \cdots, k$, $\boldsymbol{J}_{i,i}(\boldsymbol{z})=\boldsymbol{\hat B}(\boldsymbol{\hat z})\boldsymbol{B}^{-1}(\boldsymbol{z})$, and $\boldsymbol{J}(\boldsymbol{z})$ is a lower triangular matrix. Hence $\forall \boldsymbol{z} \in \mathbb{R}^{nk}$, $\boldsymbol{J}(\boldsymbol{z})$ is nonsingular, which gives that the  differentiable function $\tau= h\circ g$ is a diffeomorphism.
 \end{proof}
 
 For Theorem \ref{theorem fully controlled}, we make several remarks here:
 \begin{itemize}
 	\item Assuming $\boldsymbol{x} = g(\boldsymbol{z})$ means that all the information determining the time evolution of system (\ref{dynamic model}) is contained in the observation $\boldsymbol{x}$, which makes it possible to infer all the latents from current observations. If the observation only contains the information of $\boldsymbol{q}$, i.e., $\boldsymbol{x} = g(\boldsymbol{q})$, similar results can also be easily derived. However, in this case, we need to calculate the high-order time derivatives of $\boldsymbol{q}$ both for training and control.
 		
 	\item While we assume that $\boldsymbol{B}(\boldsymbol{z})$ is a diagonal function matrix, this can be easily extended to arbitrarily bounded and nonsingular matrix if we constrain $\boldsymbol{\hat B}(\boldsymbol{\hat z})$ to be bounded and nonsingular. However, in this case, using feedback linearization controllers need to calculate the inverse of $\boldsymbol{\hat B}(\boldsymbol{\hat z})$, which is time-consuming and may lead to instability in practice. Hence a decoupled representation and corresponding control are always preferred when it is possible.
 	
 	\item If $\boldsymbol{B}(\boldsymbol{z})$ is a (nondiagonal) constant matrix, i.e., $\boldsymbol{B}(\boldsymbol{z})=\boldsymbol{B}$, then we can always represent the dynamic system as (\ref{dynamic model}) using a linear coordinate transformation. In this case ($\boldsymbol{B}(\boldsymbol{z})$ is a diagonal constant matrix), we can 
 	 identify the latents and the model one step further, which we summarize in the following corollary. 
 \end{itemize} 
 
 \begin{corollary}
 	Suppose all the assumptions in Theorem \ref{theorem fully controlled} hold. If  $\boldsymbol{B}(\cdot)=\boldsymbol{B}$, $\boldsymbol{\hat B}(\cdot)=\boldsymbol{\hat B}$,  then $\boldsymbol{z}$ is identified up to scaling and translations. If, in addition, both $F(\cdot)$ and $\hat F(\cdot)$ are linear maps, i.e., the real dynamic system and estimated dynamic model are linear, then the coefficients in the system matrix are identified to the truth-value.
 \end{corollary}
 \begin{proof}
 	From the proof of Theorem \ref{theorem fully controlled}, we have $\boldsymbol{\hat q}_0=\tau_q(\boldsymbol{q})$ and $\boldsymbol{J}_{1,1}(\boldsymbol{q})=\boldsymbol{\hat B} \boldsymbol{B}^{-1}$ is a positive definite diagonal matrix. This immediately gives that $\tau_q(\cdot)$ is an affine function. Let $\Delta \boldsymbol{B} = \boldsymbol{\hat B} \boldsymbol{B}^{-1}$ and $\boldsymbol{\hat q}_0=\tau_q(\boldsymbol{q})=\Delta \boldsymbol{B} \boldsymbol{q}+\boldsymbol{c}$, then
 	 \begin{equation}
 		\left\{
 		\begin{array}{l}
 			 \boldsymbol{\hat q}_1=\boldsymbol{\dot {\hat q}}_0=\Delta \boldsymbol{B} \boldsymbol{\dot q}\\
 			 \qquad \vdots\\
 			 \boldsymbol{\hat q}_{k-1}=\boldsymbol{\dot {\hat q}}_{k-2}=\cdots= \boldsymbol{\hat q}^{k-1}_0=\Delta \boldsymbol{B} \boldsymbol{q}^{k-1}.
 		\end{array}
 		\right.  
 	\end{equation} 
 	And hence
 	\begin{equation}
 		 \boldsymbol{\hat z}=\begin{bmatrix} \Delta \boldsymbol{B} &\boldsymbol{0} &\cdots&\boldsymbol{0}\\\boldsymbol{0}&\Delta \boldsymbol{B}& \cdots&\boldsymbol{0}\\ \vdots&\vdots&\ddots&\vdots\\\boldsymbol{0}&\boldsymbol{0}&\cdots&\Delta \boldsymbol{B} \end{bmatrix} \boldsymbol{z} + \begin{bmatrix}  \boldsymbol{I}\\\boldsymbol{0}\\ \vdots\\\boldsymbol{0} \end{bmatrix} \boldsymbol{c},
 	\end{equation} 
 	which proves the first part of this corollary.
 	
 	If $F(\boldsymbol{z}) \doteq \boldsymbol{F}\boldsymbol{z}$ and $\hat F(\boldsymbol{\hat z}) \doteq \boldsymbol{\hat F} \boldsymbol{\hat z}$, then substituting above equation into the last $n$ rows of equation (\ref{J_z ex}) and letting $\boldsymbol{u}=\boldsymbol{0}$ give
 	\begin{equation}
 		 \boldsymbol{\hat F}\begin{bmatrix} \Delta \boldsymbol{B} &\boldsymbol{0} &\cdots&\boldsymbol{0}\\\boldsymbol{0}&\Delta \boldsymbol{B}& \cdots&\boldsymbol{0}\\ \vdots&\vdots&\ddots&\vdots\\\boldsymbol{0}&\boldsymbol{0}&\cdots&\Delta \boldsymbol{B} \end{bmatrix} \boldsymbol{z} + \boldsymbol{\hat F} \begin{bmatrix}  \boldsymbol{I}\\\boldsymbol{0}\\ \vdots\\\boldsymbol{0} \end{bmatrix} \boldsymbol{c}=   \Delta \boldsymbol{B} \boldsymbol{F}\boldsymbol{z}.
 	\end{equation} 
 	Since above equation holds for all $\boldsymbol{z} \in 
 	\mathbb{R}^{nk}$ and $\Delta \boldsymbol{B}$ is a positive definite diagonal matrix, we obtain $\boldsymbol{\hat F}=\boldsymbol{F}$, which complete the proof. And if, in addition, the first $n\times n$ matrix in $\boldsymbol{F}$ is nonsingular, we can obtain $\boldsymbol{c}=0$, which gives that $\boldsymbol{z}$ is identified up to scaling in this case.
 	
\end{proof}

 \subsubsection{Controller Design}
 \label{Sec Controller Design}
 
 In this section, we develop the controller to ensure the asymptotic tracking property based on the learned latent dynamic model. For simplicity, we use the feedback linearization method \cite{isidori1985nonlinear,hagan2002introduction,lederer2019uniform}. The controller is given as 
 
 \begin{equation}
 	\left\{
 	\begin{array}{l}
 		\boldsymbol{u}=\boldsymbol{\hat B}^{-1}( \boldsymbol{\hat z})[-\hat F(\boldsymbol{\hat z} ) + \boldsymbol{\dot {\hat q}}_{d,k-1} + \sum_{i=0}^{(k-1)}\boldsymbol{K}_i(\boldsymbol{\hat q}_{d,i} - \boldsymbol{\hat q}_{i})]  \\
 		\boldsymbol{\hat z}_d\doteq \begin{bmatrix}\boldsymbol{\hat q}_{d,0}\\ \vdots \\ \boldsymbol{\hat q}_{d,k-1}   \end{bmatrix}=h(\boldsymbol{x}_d),
 	\end{array}
 	\right. 
 	\label{controller 1}
 \end{equation} 
 where $\boldsymbol{\hat z}$ (and hence $\boldsymbol{\hat q}_i$) is inferred from $\boldsymbol{\hat x}$ by using learned representation function (\ref{h}). Observation space trajectory is given by $\boldsymbol{x}_d=g(\boldsymbol{z}_d)$, where $\boldsymbol{z}_d \doteq [\boldsymbol{q}_d^T, \boldsymbol{\dot q}_d^T, \cdots, (\boldsymbol{q}_d^{(k-1)})^T]^T$. In practice, $\boldsymbol{\dot {\hat q}}_{d,k-1}$ is approximated
  by $ [ \boldsymbol{\boldsymbol{\hat q}}_{d,k-1}(t+\Delta t)-\boldsymbol{\boldsymbol{\hat q}}_{d,k-1}(t) ]\big /\Delta t$. $\boldsymbol{K}_0,\cdots,\boldsymbol{K}_{k-1}$ are manually designed controller coefficient matrices.
  
 For simplicity we define 
  \begin{equation}  \boldsymbol{A}=\begin{bmatrix}0&1&0&\cdots&0\\ 0&0&1&\cdots&0 \\ \vdots&\vdots&\vdots&\ddots&\vdots\\0&0&0&\cdots&1\\ -\boldsymbol{K}_0&-\boldsymbol{K}_1&-\boldsymbol{K}_2&\cdots&-\boldsymbol{K}_{k-1} \end{bmatrix},
 	\label{A}
 \end{equation}
 then the following conclusion holds.
  \begin{theorem}
 	Suppose the latent dynamic system is given by (\ref{dynamic model}) and (\ref{g}).  The estimated latent dynamic model given in (\ref{h}) and (\ref{estimated dynamic model}) is learned by solving (\ref{minimizing}). If all eigenvalues of $\boldsymbol{A}$ in (\ref{A}) have negative real parts, then the controller (\ref{controller 1}) ensures that the real tracking error $\boldsymbol{e}\to \boldsymbol{0}$ as $t \to \infty$.
 	\label{theorem tracking}
 \end{theorem}	
 
	\begin{proof}
		 Let $\boldsymbol{\hat e}=\boldsymbol{\hat z}_d-\boldsymbol{\hat z}$. Theorem \ref{theorem fully controlled} ensures that $\boldsymbol{\dot {\hat z}}= \hat {\mathcal{F}}(\boldsymbol{\hat z}, \boldsymbol{u})$, and also $\forall i=0,\cdots,k-2$, $\boldsymbol{\dot {\hat q}}_{d,i}=\boldsymbol{\hat q}_{d,i+1}$. Then, the estimated closed-loop system is given as 
		\begin{equation}
			\boldsymbol{\dot {\hat e}} =\boldsymbol{A}\boldsymbol{\hat e}.
			\label{estimated closed-loop system}
		\end{equation}
		The estimated closed-loop system (\ref{estimated closed-loop system}) is a linear system with all eigenvalues of $\boldsymbol{A}$ having negative real parts. Hence for all initial states, the estimated tracking error $\boldsymbol{\hat z}_d-\boldsymbol{\hat z}=\boldsymbol{\hat e} \to \boldsymbol{0}$ as $t \to \infty$. Then since $\boldsymbol{\hat z}=\tau(\boldsymbol{z})$ and $\boldsymbol{\hat z}_d=\tau(\boldsymbol{z}_d)$ with $\tau(\cdot)$ being a diffeomorphism (see, Theorem \ref{theorem fully controlled}), the real tracking error $\boldsymbol{e}=\boldsymbol{z}_d-\boldsymbol{z}\to \boldsymbol{0}$ as $t \to \infty$. 
	\end{proof}

	If we consider regulation control to a desired state $\boldsymbol{z}_d \doteq [\boldsymbol{q}_d^T, \boldsymbol{0}, \cdots, \boldsymbol{0}]^T$, one can simply let $\forall i=1,\cdots,k-1$, $\boldsymbol{\hat q}_{d,i}=\boldsymbol{0}$ and $\boldsymbol{\dot {\hat q}}_{d,k-1}=\boldsymbol{0}$ in the controller (\ref{controller 1}). It is easy to verify that as $t \to \infty$, $\boldsymbol{\hat q}_d-\boldsymbol{\hat q} \to \boldsymbol{0}$. And since Theorem \ref{theorem fully controlled} ensures that $\boldsymbol{\hat q}_0=\tau_q(\boldsymbol{q})$ and $\tau_q(\cdot)$ is a diffeomorphism,  $\boldsymbol{q}_d-\boldsymbol{q}$ also tends to $\boldsymbol{0}$ as $t \to \infty$.

 \subsection{Extensions to Systems with Uncontrollable Latents}
 \label{Sec Extensions to Systems with Uncontrollable Latents}
 
 In complex environments, the observations may contain information of both controllable and uncontrollable components. If the uncontrollable components are invariant in the dataset, their information will be modeled in the representation function and hence will not be mixed with the estimated latents. However, if they are factors of variance (FoV) \cite{bengio2013representation} in the dataset, the estimated controllable latents may be entangled with them, which is undesirable. 
 
 These uncontrollable latents may be statical or dynamic. For statical latents, one can use methods in, e.g., \cite{von2021self,von2023nonparametric,ahuja2024multi}, as a pre-training method to isolate the controllable latents from the uncontrollable ones under suitable assumptions therein. In this work, we only deal with the dynamical uncontrollable latents denoted as $\boldsymbol{s} \in \mathbb{R}^l$. This motivation is similar to works considering exogenous block Markov decision process (EX-BMDP) \cite{efroni2021provably,lamb2022guaranteed,levine2024multistep}, but the settings are different in the sense that we consider dynamic systems with continuous state and action spaces (instead of MDP with finite state and action sets), which makes it suitable for continuous control.

  In this case, we need multiple environments where the dynamics of the uncontrollable latents are different. To be specific, we consider there are $l+1$ different training environments. The dynamics of the controllable latents $\boldsymbol{z}$ are the same (see equation (\ref{dynamic model})) across all the environments. $\forall i=1,\cdots,l+1$, the dynamics of the uncontrollable latent $\boldsymbol{s}$ in the $i$-th environment is
 \begin{equation} 
 		\boldsymbol{\dot s}=G_{i}(\boldsymbol{s}). 
 	\label{uncontrolled dynamics}
 \end{equation}  
 
 The dynamics of $\boldsymbol{s}$ in the test environment is 
  \begin{equation} 
 	\boldsymbol{\dot s}=G_{\mbox {test}}(\boldsymbol{s},t),
 	\label{uncontrolled dynamics test}
 \end{equation}  
 which is assumed to be stable.
 
Now the mixing function $g_{zs}: \mathbb{R}^{nk+l}\to \mathbb{R}^{m}$ is a function of both $\boldsymbol{z}$ and $\boldsymbol{s}$, i.e.,
 \begin{equation} 
 	\boldsymbol{x}=g_{zs}(\boldsymbol{z},\boldsymbol{s}),
 	\label{g zs}
 \end{equation}  
 where $g_{zs}(\cdot)$ is also assumed to be $k$ times differentiable and injective. And the expectation in (\ref{minimizing}) is now taken over  all the states $\boldsymbol{x}\in \mathcal{X}$, generated by all $\boldsymbol{z} \in  \mathbb{R}^{nk}$ and $\boldsymbol{s} \in \mathbb{R}^l$ through $g_{zs}(\cdot)$, and all $\boldsymbol{u} \in \mathbb{R}^{n}$.
 
  All other settings are the same as in the former section, and the following conclusions hold.
 
 \begin{theorem}
 	Suppose the latent dynamic system is given by (\ref{dynamic model}), (\ref{uncontrolled dynamics}) and (\ref{g zs}), and $\forall \boldsymbol{s} \in \mathbb{R}^{l}$, the matrix $[G_{2}(\boldsymbol{s})-G_{1}(\boldsymbol{s}),\cdots,G_{l+1}(\boldsymbol{s})-G_{1}(\boldsymbol{s})]$ is nonsingular. If we learn a model (\ref{h}) and (\ref{estimated dynamic model}) by solving (\ref{minimizing}) in all training environments, then the following conclusions hold.
 	\begin{itemize}
 	\item 1. $\boldsymbol{q}$ is identified up to an invertible componentwise transformation and $\tau_{zs}\doteq h\circ g_{zs}$ is a diffeomorphism that only depends on $\boldsymbol{z}$.
 	\item 2. If all eigenvalues of $\boldsymbol{A}$ in (\ref{A}) have negative real parts, then controller (\ref{controller 1}) ensures that the real tracking error $\boldsymbol{e}\to \boldsymbol{0}$ as $t \to \infty$ in the test environment given by (\ref{dynamic model}), (\ref{uncontrolled dynamics test}) and (\ref{g zs}).
 	\end{itemize}  
 	\label{theorem partially controlled}
 \end{theorem}	
 
 \begin{proof}
 	 Letting $\boldsymbol{\hat z} =h\circ g_{zs} (\boldsymbol{z},\boldsymbol{s})= \tau_{zs}(\boldsymbol{z},\boldsymbol{s})$ and taking the time derivative of it give
 \begin{align}
 	\boldsymbol{\dot {\hat z}}=&\boldsymbol{J}_{zs}(\boldsymbol{z},\boldsymbol{s})\begin{bmatrix}\boldsymbol{\dot z}\\\boldsymbol{\dot s} \end{bmatrix} \\
 	\boldsymbol{\dot {\hat z}}=& \boldsymbol{J}_z(\boldsymbol{z},\boldsymbol{s})\boldsymbol{\dot z} + \boldsymbol{J}_s(\boldsymbol{z},\boldsymbol{s})\boldsymbol{\dot s}\\ 
 	\boldsymbol{\dot {\hat z}}=& \boldsymbol{J}_z(\boldsymbol{z},\boldsymbol{s})\boldsymbol{\dot z} + \boldsymbol{J}_s(\boldsymbol{z},\boldsymbol{s}) G_{i}(\boldsymbol{s}),
 \end{align}
 where $\boldsymbol{J}_{zs}(\boldsymbol{z},\boldsymbol{s})\doteq [\boldsymbol{J}_z(\boldsymbol{z},\boldsymbol{s}),\boldsymbol{J}_s(\boldsymbol{z},\boldsymbol{s})]$ is the Jacobian of $\tau_{zs}$.
 
 Since above equation holds for all $i=1,\cdots,l+1$, then we have
 \begin{align} 
 	\boldsymbol{J}_s(\boldsymbol{z},\boldsymbol{s}) [G_{2}(\boldsymbol{s})-G_{1}(\boldsymbol{s}),\cdots,G_{l+1}(\boldsymbol{s})-G_{1}(\boldsymbol{s})]=\boldsymbol{0}.
 \end{align}
 
 Since by assumption, the matrix $[G_{2}(\boldsymbol{s})-G_{1}(\boldsymbol{s}),\cdots,G_{k+1}(\boldsymbol{s})-G_{1}(\boldsymbol{s})]$ is nonsingular, we have $\boldsymbol{J}_s(\boldsymbol{z},\boldsymbol{s})=\boldsymbol{0}$, which means $\boldsymbol{\hat z}= \tau_{zs}(\boldsymbol{z},\boldsymbol{s})$ is not a function of $\boldsymbol{s}$, and hence is only a function of $\boldsymbol{z}$. Then, we can derive conclusion 1 by following similar procedures as in the proof of Theorem \ref{theorem fully controlled}.
 
 Since in the test environment, the dynamics of $\boldsymbol{s}$ are assumed to be stable, and both $\boldsymbol{\hat z}$ and $\boldsymbol{z}$ are independent of $\boldsymbol{s}$, we can derive conclusion 2 by following similar procedures as in the proof of Theorem \ref{theorem tracking}.
 \end{proof}  
\section{Results}

 In this section, we provide the experimental results based on an attitude dynamic model of a rigid-body spacecraft.
 
\subsection{Experimental Setup}
 The experiments are conducted by first generating training data according to Section \ref{Latent Dynamic Models and Data Generation}. Then, train the estimated latent dynamic models as in Section \ref{Estimated Latent Dynamic Models}. Finally, use the controller in Section \ref{Controllers} to ensure the tracking property of closed-loop systems.
 \subsubsection{Latent Dynamic Models and Data Generation}
 \label{Latent Dynamic Models and Data Generation}
  We consider the following latent dynamic model
   \begin{equation}
  	\left\{
  	\begin{array}{l}
  		\dot \theta_x=\omega_x\\
  		\dot \theta_y=\omega_y\\
  		\dot \theta_z=\omega_z\\
  		J_x\dot \omega_x=\left(J_y-J_z\right)\omega_y\omega_z+T_x\\
  		 J_y\dot \omega_y=\left(J_z-J_x\right)\omega_z\omega_x+T_y\\
  		J_z \dot \omega_z=\left(J_x-J_y\right)\omega_x\omega_y+T_z,
  	\end{array}
  	\right.  
  	\label{attitude mod}
  \end{equation} 
  with
  \begin{equation}
  	\left\{
  	\begin{array}{l}
  		\boldsymbol{x}=g(\boldsymbol{z})\\
  		\boldsymbol{z}=[\theta_x,\theta_y,\theta_z,\omega_x,\omega_y,\omega_z]^T\\
  		\boldsymbol{u}=[T_x,T_y,T_z]^T.
  	\end{array}
  	\right.  
  \end{equation} 
   
  This model is adapted from a standard attitude dynamic model of a rigid-body spacecraft \cite[p.169]{xie2022spacecraft}. The parameters are set to $J_x=0.8$, $J_y=1.0$, $J_z=1.2$.
  
   Mixing function $g$ is approximated by a randomly initialized  multi-layer perceptron (MLP) with two hidden layers by following \cite{zimmermann2021contrastive, ahuja2022weakly}. The hidden dimensions are the same as the input dimension, and the activation functions for hidden layers are SmoothLeakyReLU $\sigma(x)=0.2x+0.8\log(1+e^x)$.
  { The output dimension is 50, and hence $\boldsymbol{x} \in \mathbb{R}^{50}$.}
  
  We let each component of the initial state  $\boldsymbol{z}(t_0)$ and $\boldsymbol{u}(t_0)$ follows an uniform distribution $U(-1,1)$, and calculate the next-step state $\boldsymbol{z}(t_0+\Delta t)$ according to (\ref{attitude mod}). { Then calculate $\boldsymbol{x}(t_0)=g(\boldsymbol{z}(t_0))$ and $\boldsymbol{x}(t_0+\Delta t)=g(\boldsymbol{z}(t_0+\Delta t))$.} The sampling period $\Delta t$ is set as $0.01$(s), and we use the classic Runge–Kutta method (RK4) for numerical simulations. 
  
   Both in training and control processes, the dynamic model (\ref{attitude mod}) and nonlinear mixing function $g$ are unknown, and  $\boldsymbol{z}=[\theta_x,\theta_y,\theta_z,\omega_x,\omega_y,\omega_z]^T$ is not accessible.  Only the observation $\boldsymbol{x} \in \mathbb{R}^{50}$ and $\boldsymbol{u}=[T_x,T_y,T_z]^T$ are attainable. We generate 300,000 samples ($\boldsymbol{x}(t_0),\boldsymbol{x}(t_0+\Delta t),\boldsymbol{u}(t_0)$) for training.
   
   {We first provide the results for the controllable system without noises described above in Section \ref{Results for Controllable Systems without Noises}. Then we consider the effect of observation noises, control deviations, and uncontrollable latents. 
   
   For systems with observation noises, two different noises, namely fast-varying noises and slow-varying noises, are considered. For fast-varying noises, the observations $\boldsymbol{x}(t_0)$ and $\boldsymbol{x}(t_0+\Delta t)$ are perturbed by adding zero mean Gaussian noises $\tilde{\boldsymbol{x}}_0 \in \mathbb{R}^{50}$ and $\tilde{\boldsymbol{x}}_1 \in \mathbb{R}^{50}$, respectively.  For slow-varying noises, the noise variables are the same during the sampling time $\Delta t=0.01$s, i.e., $\tilde{\boldsymbol{x}}_0= \tilde{\boldsymbol{x}}_1$. These noises are added both in the training dataset and during the control processes. The results are given in Section \ref{Results for Controllable Systems with Observation Noises}.

   For systems with control deviations, we perturb $\boldsymbol{u}$ by adding a noise variable $\tilde{\boldsymbol{u}} \in \mathbb{R}^{3}$, where each component of $\tilde{\boldsymbol{u}}$ follows the zero mean Gaussian distribution. The results are given in Section \ref{Results for Controllable Systems with Control Deviations}.

   For systems with uncontrollable latents, we consider another one-dimensional latent ${s} \in \mathbb{R}$. In two different training environment, the dynamics of ${s}$ are given as ${\dot s} = 5-{s}$ and ${\dot s} =-5+5 {s}^2 $. In the test environment, the dynamics of ${s}$ is given as  $ {s} = -0.5\cos(t)$. And in this case, the mixing function $g(\boldsymbol{z})$ is replaced by $g_{zs}(\boldsymbol{z},{s})$. All other settings remain the same. The results are given in Section \ref{Results for Systems with Uncontrollable Latents}.}

  \subsubsection{Estimated Latent Dynamic Models}
  \label{Estimated Latent Dynamic Models}
  The estimated latent dynamic model is given in (\ref{h}) and (\ref{estimated dynamic model}). We use a 3-layer fully connected neural network as the representation function $h$, and the activation functions for hidden layers are Leaky-ReLU with a negative slope of 0.2. $\hat F(\cdot)$ and $\boldsymbol{\hat B}(\cdot)$ are approximated by 3-layer fully connected neural networks with Leaky-ReLU (0.2) as activation functions of hidden layers. For $\boldsymbol{\hat B}(\cdot)$, we use ReLU with a bias of 0.1 as activation functions for the output layer.
  The model is learned by solving (\refeq{minimizing}), with $\boldsymbol{\dot {\hat z}}$ approximated by using (\ref{forward difference}). 
  
  \subsubsection{Controllers}
  \label{Controllers}
  The controller is given by (\ref{controller 1}), with $\boldsymbol{K}_0=50\boldsymbol{I}_3$ and $\boldsymbol{K}_1=50\boldsymbol{I}_3$. The reference $\boldsymbol{x}_d$ is generated by
  \begin{equation}
  	\left\{
  	\begin{array}{l}
  		\boldsymbol{x}_d=g(\boldsymbol{z}_d)\\
  		\boldsymbol{z}_d=[r_x,r_y,r_z,\dot r_x,\dot r_y,\dot r_z]^T\\
  		\hspace{13pt}=[0.8-0.4e^{-t},0.6-1.2e^{-t},0.6+0.3\sin(t),0.4e^{-t},1.2e^{-t},0.3\cos(t)]^T
  	\end{array}
  	\right.  
  \end{equation}

     \subsection{Experimental Results} 
      
      \subsubsection{Results for Controllable Systems without Noises}
      \label{Results for Controllable Systems without Noises}
      The tracking property for the fully controllable system is given in Figure \ref{fig:0}, and the loss function curve is reported in Figure \ref{loss function curve}.
       \begin{figure}[h]
      	\centering
      	\begin{subfigure}{0.4\textwidth}
      		\includegraphics[scale=0.23]{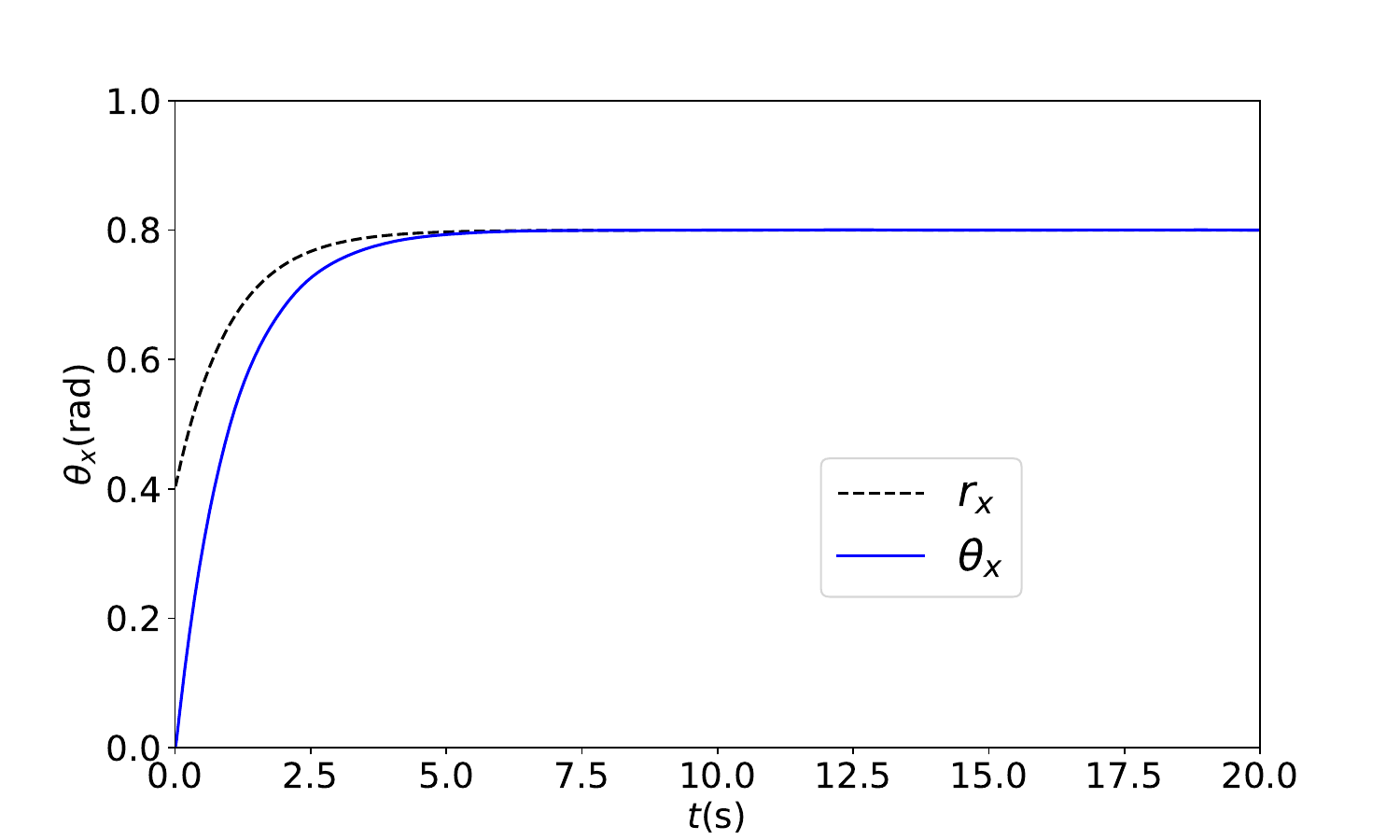}
      		\caption{\label{fig:01}}
      	\end{subfigure}
      	\begin{subfigure}{0.4\textwidth}
      		\includegraphics[scale=0.23]{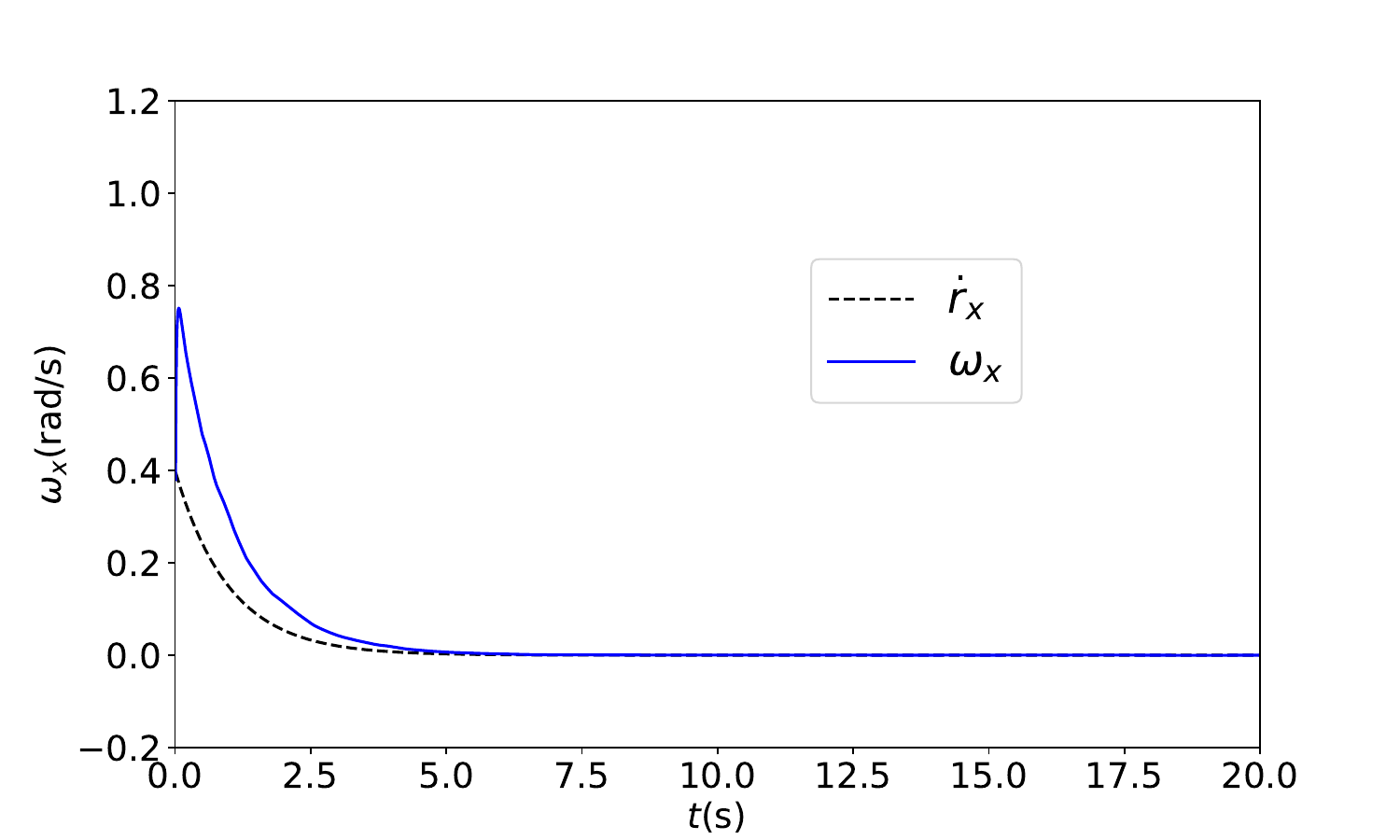}
      		\caption{\label{fig:04}}
      	\end{subfigure}
      	\begin{subfigure}{0.4\textwidth}
      		\includegraphics[scale=0.23]{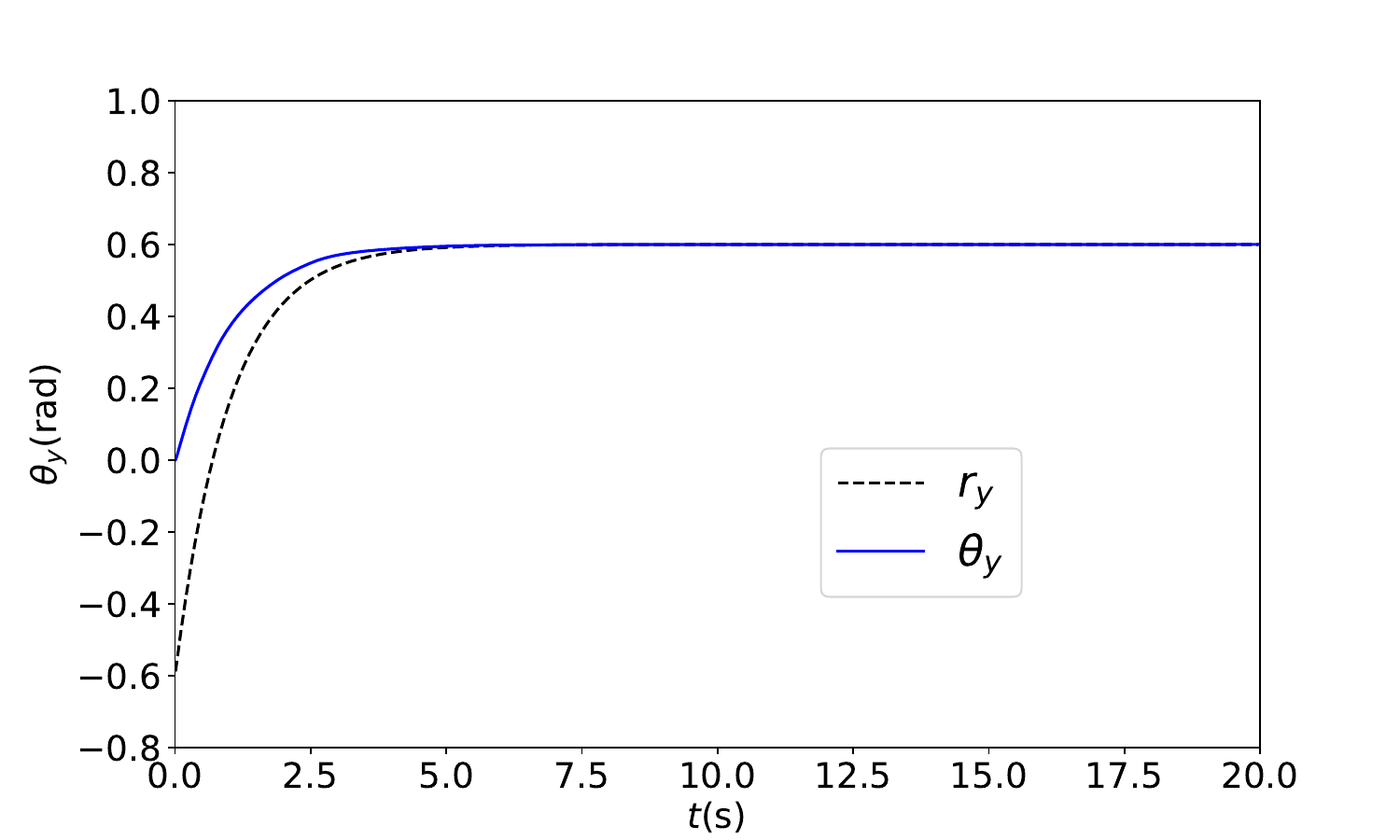}
      		\caption{\label{fig:02}}
      	\end{subfigure}
      	\begin{subfigure}{0.4\textwidth}
      		\includegraphics[scale=0.23]{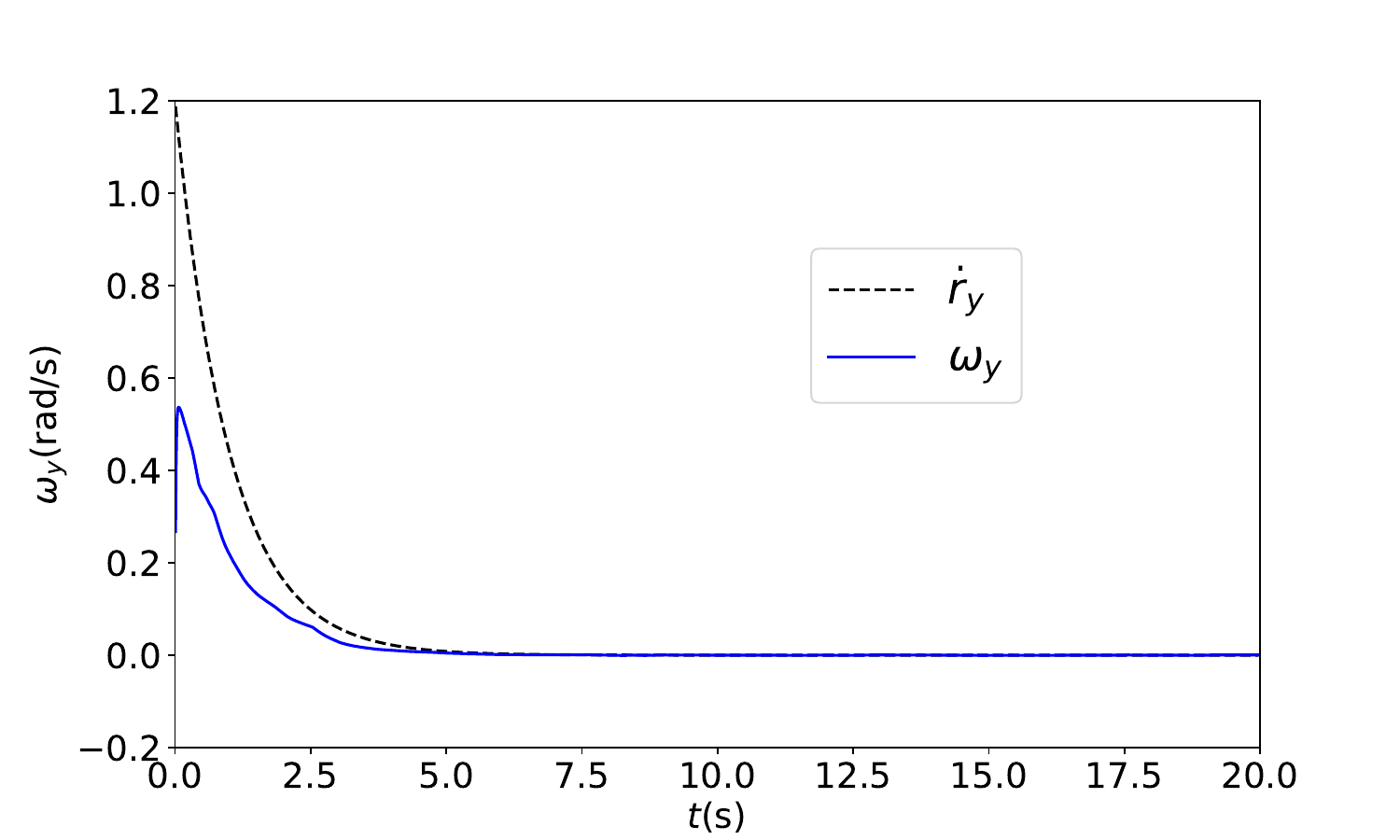}
      		\caption{\label{fig:05}}
      	\end{subfigure}
      	\begin{subfigure}{0.4\textwidth}
      		\includegraphics[scale=0.23]{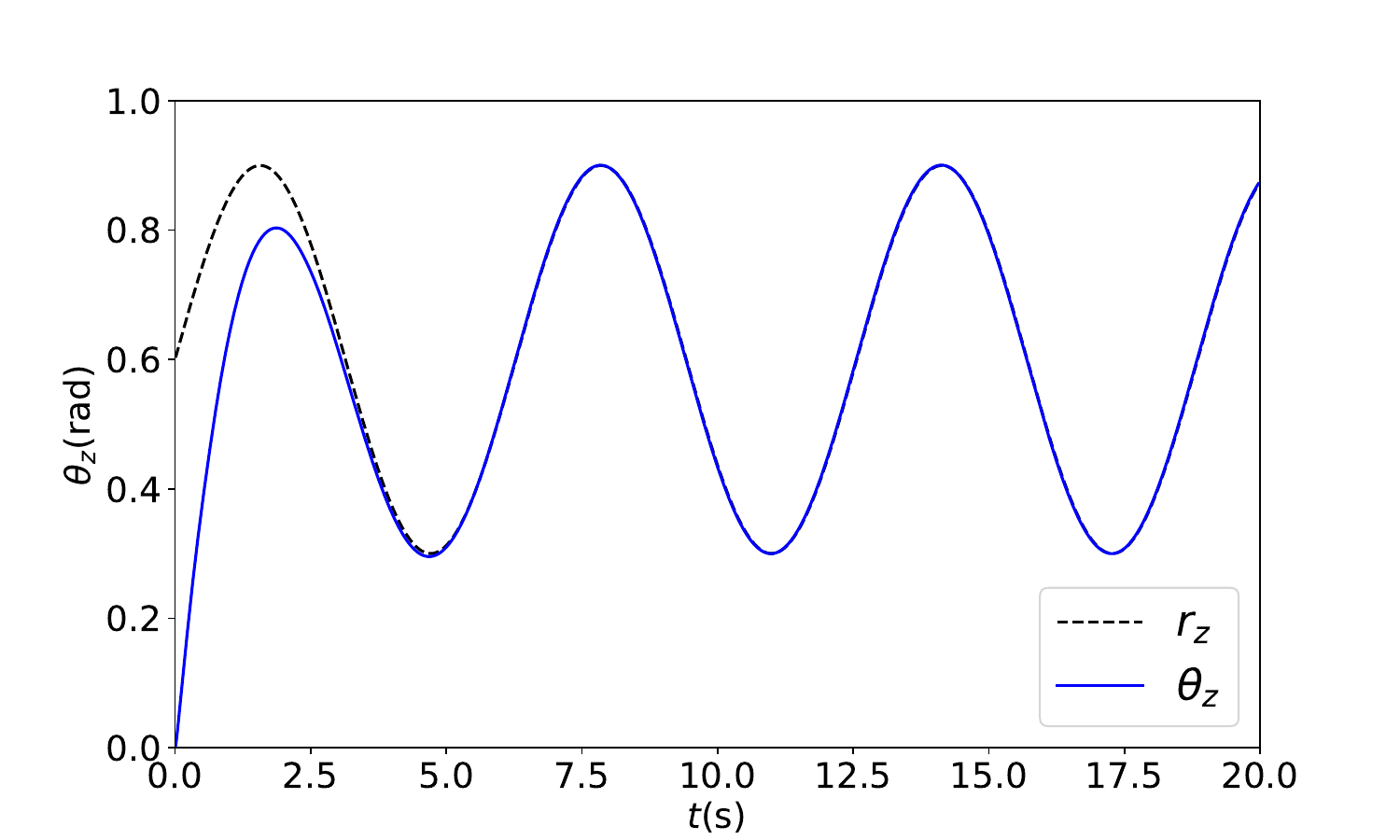}
      		\caption{\label{fig:03}}
      	\end{subfigure}
      	\begin{subfigure}{0.4\textwidth}
      		\includegraphics[scale=0.23]{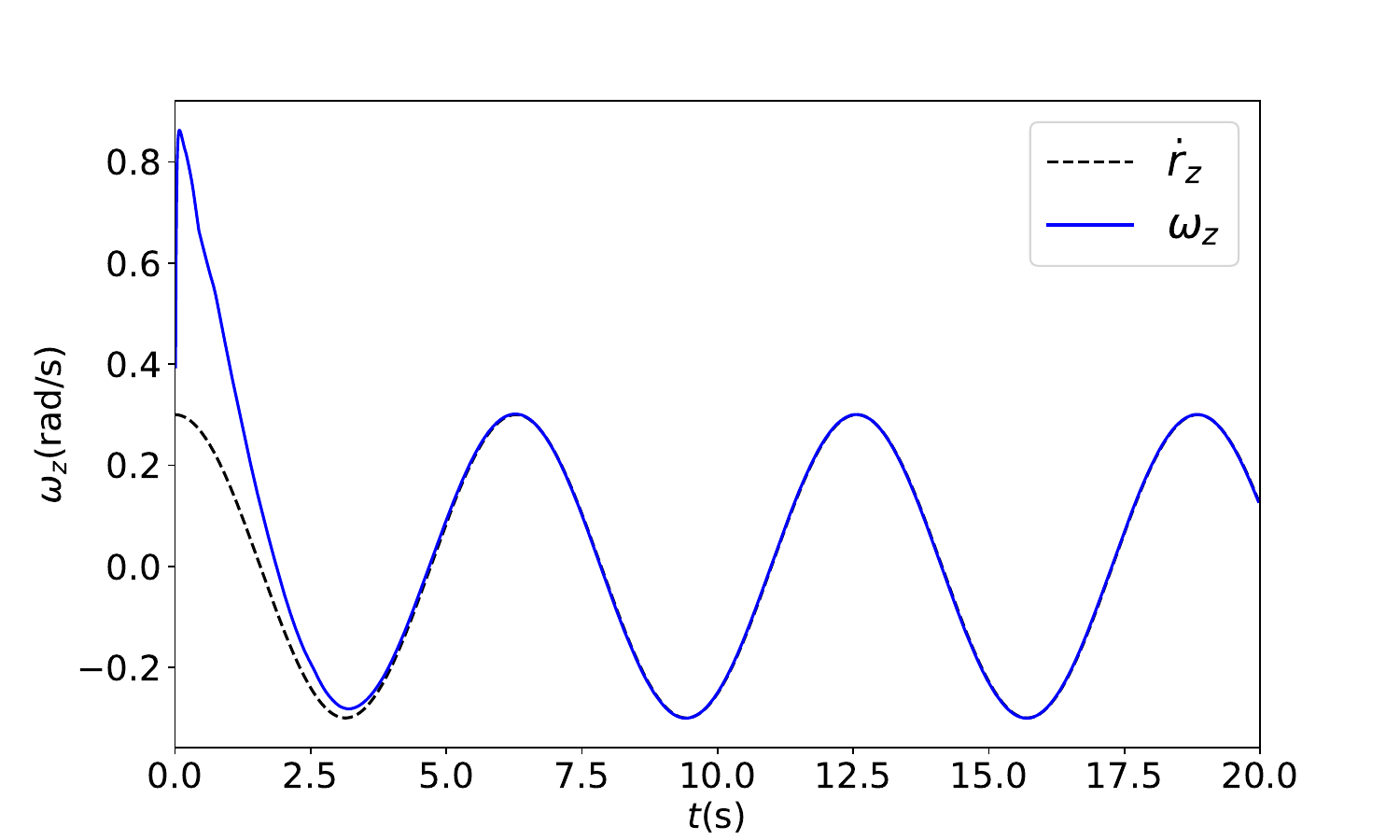}
      		\caption{\label{fig:06}}
      	\end{subfigure}
      	\caption{Tracking property for the fully controllable system. Time evolution of: (\subref{fig:01}) $\theta_x$, (\subref{fig:04}) $\omega_x$, (\subref{fig:02}) $\theta_y$ (\subref{fig:05}) $\omega_y$, (\subref{fig:03}) $\theta_z$, (\subref{fig:06}) $\omega_z$.}
      	\label{fig:0}
      \end{figure} 
      
      \begin{figure}[h]
      	\centering 
      	\includegraphics[scale=0.3]{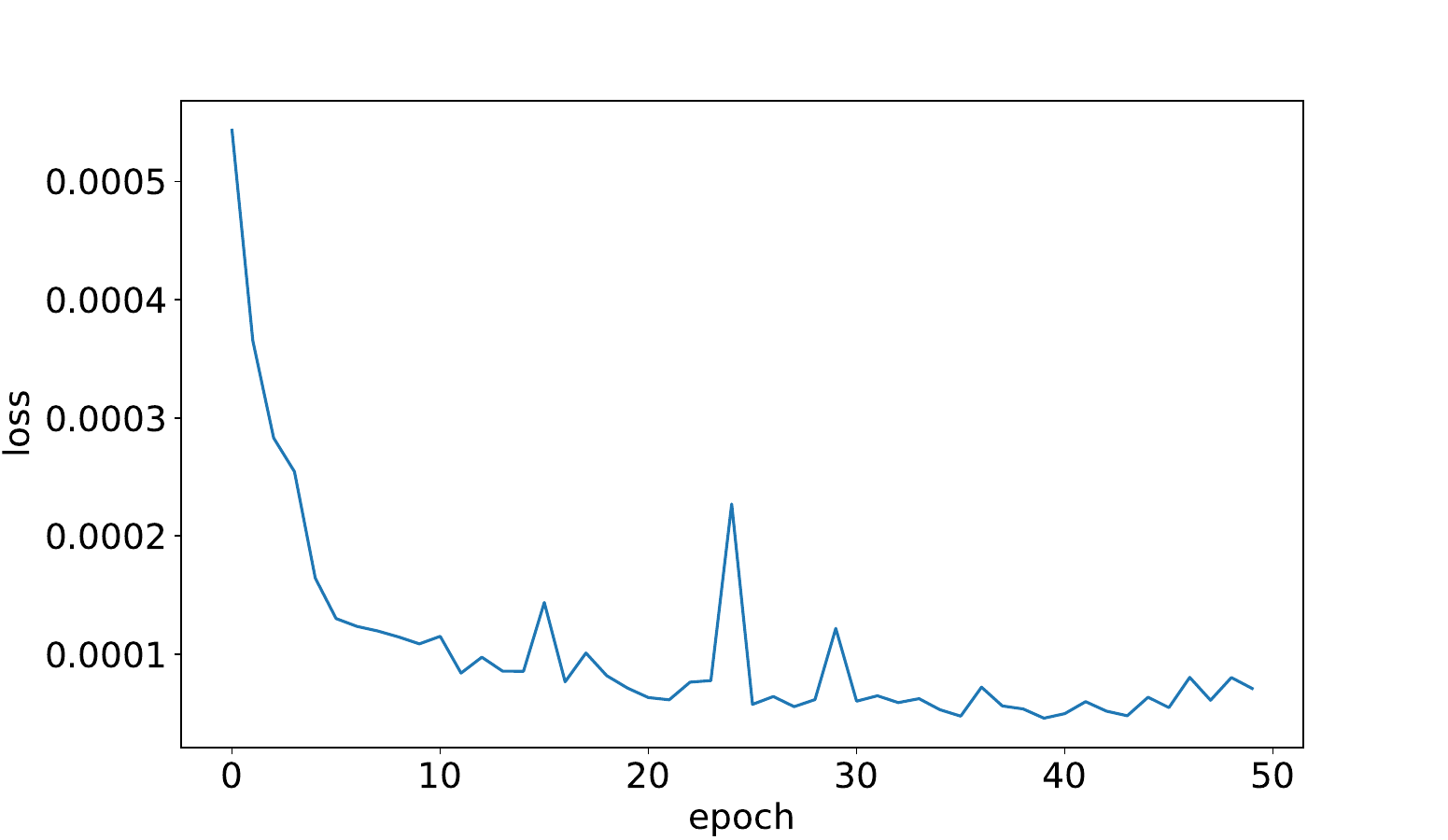} 
      	\caption{loss function curve} 
      	\label{loss function curve}
      \end{figure}
            
       {The convergence time is 5.0 (s), and there is almost no overshoot. We report the tracking errors based on the infinite norm of the error signal in the period $t \in [10,20]$ (s). The attitude tracking error is $[0.0006,0.0004,0.0029]$ (rad), and the attitude angular velocity tracking error is $[0.0007,0.0009,0.0034]$ (rad). The results show that the proposed controller ensures that the closed-loop system tracks the given reference signal well.}
      
     {
	\subsubsection{Results for Controllable Systems with Observation Noises}
  	\label{Results for Controllable Systems with Observation Noises}
  	  
  	  We report the attitude tracking error in Table \ref{Table Results for Controllable Systems with Observation Noises}. 
   \begin{table}[h]
  	\caption{Results for Controllable Systems with Observation Noises.}
  	\label{Table Results for Controllable Systems with Observation Noises}
  	\centering
  	\begin{tabular}{cccc}
  		\hline
  		$\sigma_{x,train}$&Noise Types&$\sigma_{x,control}$&Tracking Errors ($ rad$)\\
  		\hline
  		0 & - &0.001 &[0.0017,\ 0.0010,\ 0.0039] \\ 
  		0 & - &0.01 &[0.0112,\ 0.0112,\ 0.0191] \\ 
  		0 & - &0.1 &[0.1780,\ 0.2176,\ 0.1585] \\
  		0.01 & slow-varying &0 &[0.0004,\ 0.0004,\ 0.0029] \\
  		0.1 & slow-varying  &0 &[0.0081,\ 0.0130,\ 0.0100] \\
  		0.1 & slow-varying  &0.1 &[0.0699,\ 0.2641,\ 0.2188] \\
  		0.001 & fast-varying &0 &[0.0067,\ 0.0113,\ 0.0144] \\ 
  		0.01& fast-varying &0 &unstable \\ 
  		0.01& fast-varying &0.1 &unstable \\ 
  		\hline
  	\end{tabular}  
  \end{table}
  
   In Table \ref{Table Results for Controllable Systems with Observation Noises}, $\sigma_{x,train}$ and $\sigma_{x,control}$ denote the standard deviation of the observation noises in the training and control processes, respectively. Noise types indicate the type of the training observation noises, while the type of the noises in the control process is always fast-varying. The results show that slow-varying training observation noises have little effect on the control performance while fast-varying training observation noises are easy to destabilize the system. One possible explanation is that the variation of fast-varying noises covers the effect of control during the training process. The noises in the control process also have some effect on the control performance since they can be regarded as a kind of measurement noise.
  
	\subsubsection{Results for Controllable Systems with Control Deviations}
  	\label{Results for Controllable Systems with Control Deviations}
     The attitude tracking errors for controllable systems with control deviations are given in Table \ref{Table Results for Controllable Systems with Control Deviations}, where $\sigma_{u,train}$ and $\sigma_{u,control}$ denote the standard deviation of $\tilde{\boldsymbol{u}}$ in the training and control processes, respectively.
  
   \begin{table}[h]
   	\caption{Results for Controllable Systems with Control Deviations.}
   	\label{Table Results for Controllable Systems with Control Deviations}
   	\centering
   	\begin{tabular}{cccc}
   		\hline
   		$\sigma_{u,train}$&$\sigma_{u,control}$&Tracking Errors ($ rad$)\\
   		\hline 
   		0 &  0.1 &[0.0008,\ 0.0007,\ 0.0030] \\ 
   		0 &  1 &[0.0053,\ 0.0031,\ 0.0042] \\
   		0.1 &  0 &[0.0007,\ 0.0007,\ 0.0028] \\
   		1 &  0 &[0.0189,\ 0.0065,\ 0.0183] \\ 
   		0.1 &  0.1 &[0.0010,\ 0.0008,\ 0.0030] \\ 
   		1 &  1 &[0.0225,\ 0.0157,\ 0.0216] \\ 
   		\hline
   	\end{tabular}  
   \end{table}
   The results show that the noises in control channels will affect the control performance, but limited control deviations will not make the system unstable, which means the closed-loop system has some degree of robustness to the control deviation.}
 
   \subsubsection{Results for Systems with Uncontrollable Latents}
   \label{Results for Systems with Uncontrollable Latents}
   For systems with uncontrollable latents, the results for training in a single environment and two different environments are given in Figure \ref{fig:1} and Figure \ref{fig:2}, respectively. For the experiment in a single environment, we generate 600,000 samples for training using (\ref{attitude mod}) and ${\dot s} = 5-{s}$. For the experiment in two environments, 300,000 training samples are generated in each environment.

 \begin{figure}[h]
	\centering
	\begin{subfigure}{0.4\textwidth}
		\includegraphics[scale=0.23]{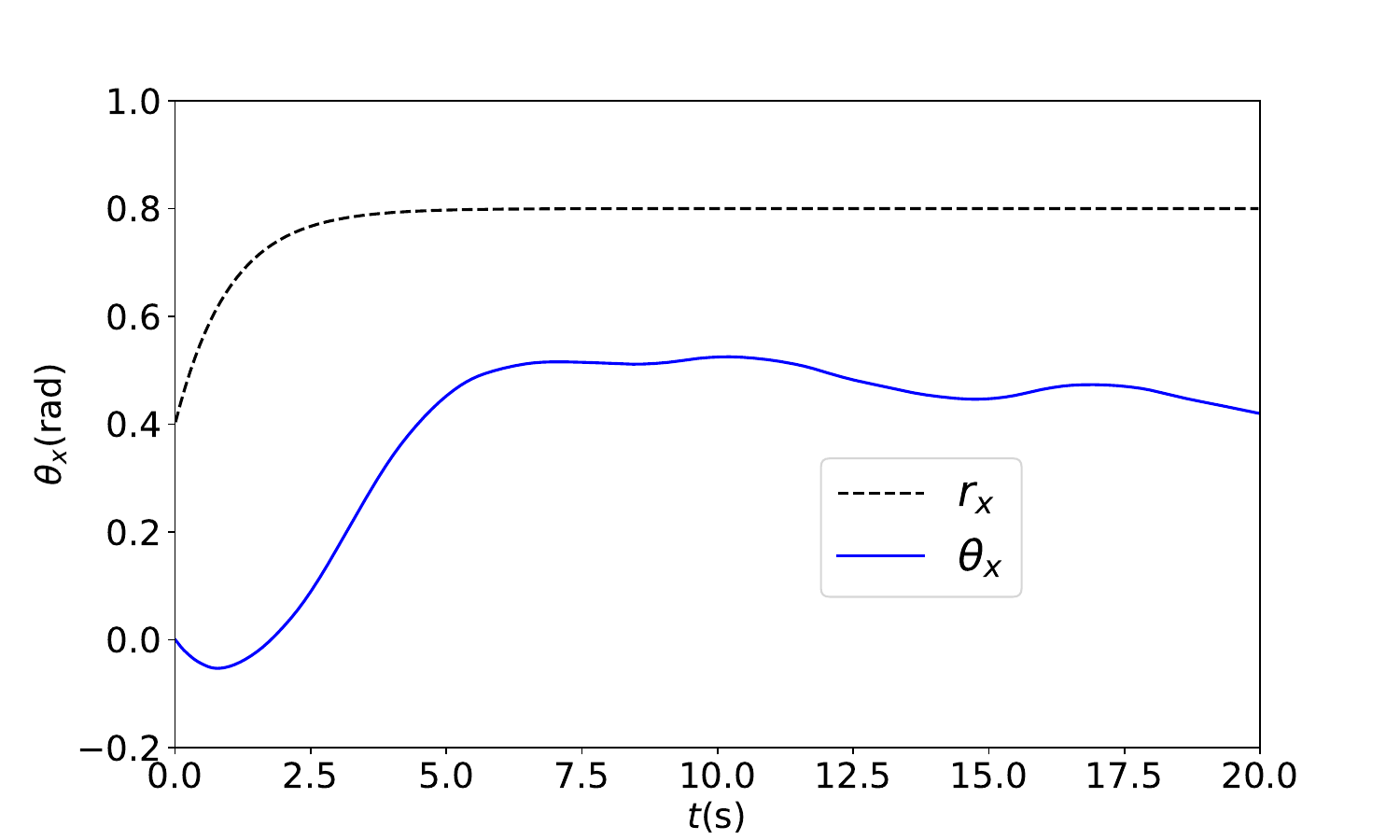}
		\caption{\label{fig:11}}
	\end{subfigure}
	\begin{subfigure}{0.4\textwidth}
		\includegraphics[scale=0.23]{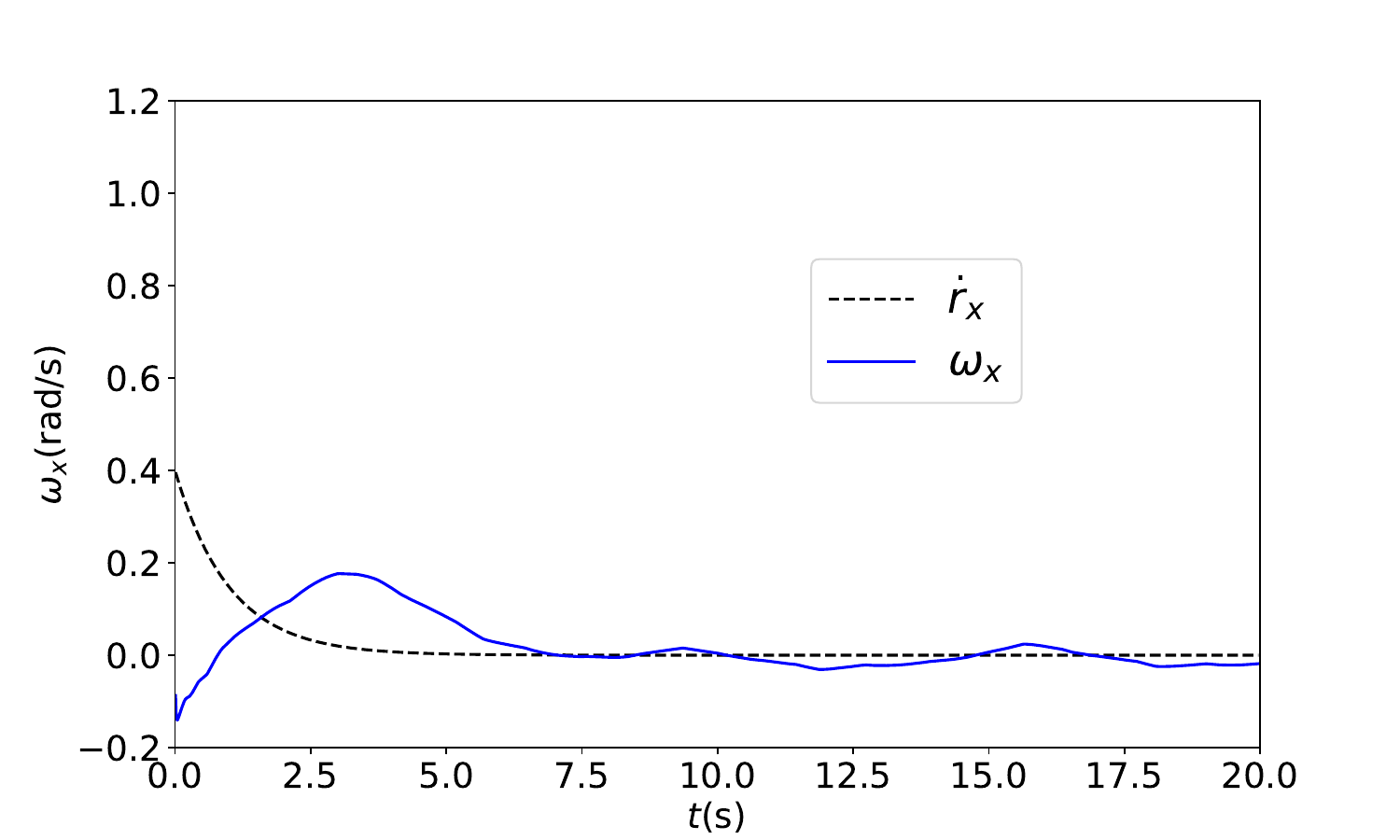}
		\caption{\label{fig:14}}
	\end{subfigure}
	\begin{subfigure}{0.4\textwidth}
		\includegraphics[scale=0.23]{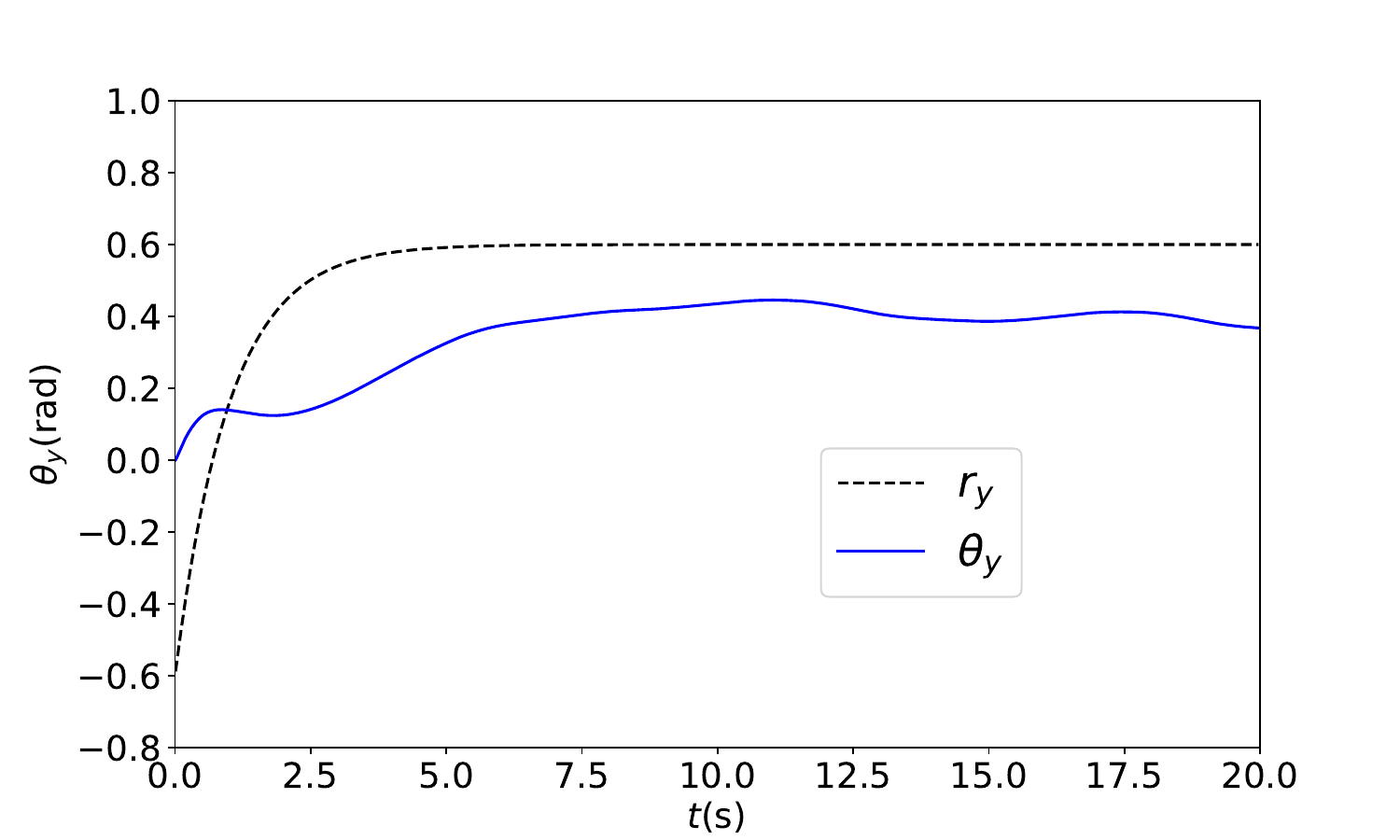}
		\caption{\label{fig:12}}
	\end{subfigure}
	\begin{subfigure}{0.4\textwidth}
		\includegraphics[scale=0.23]{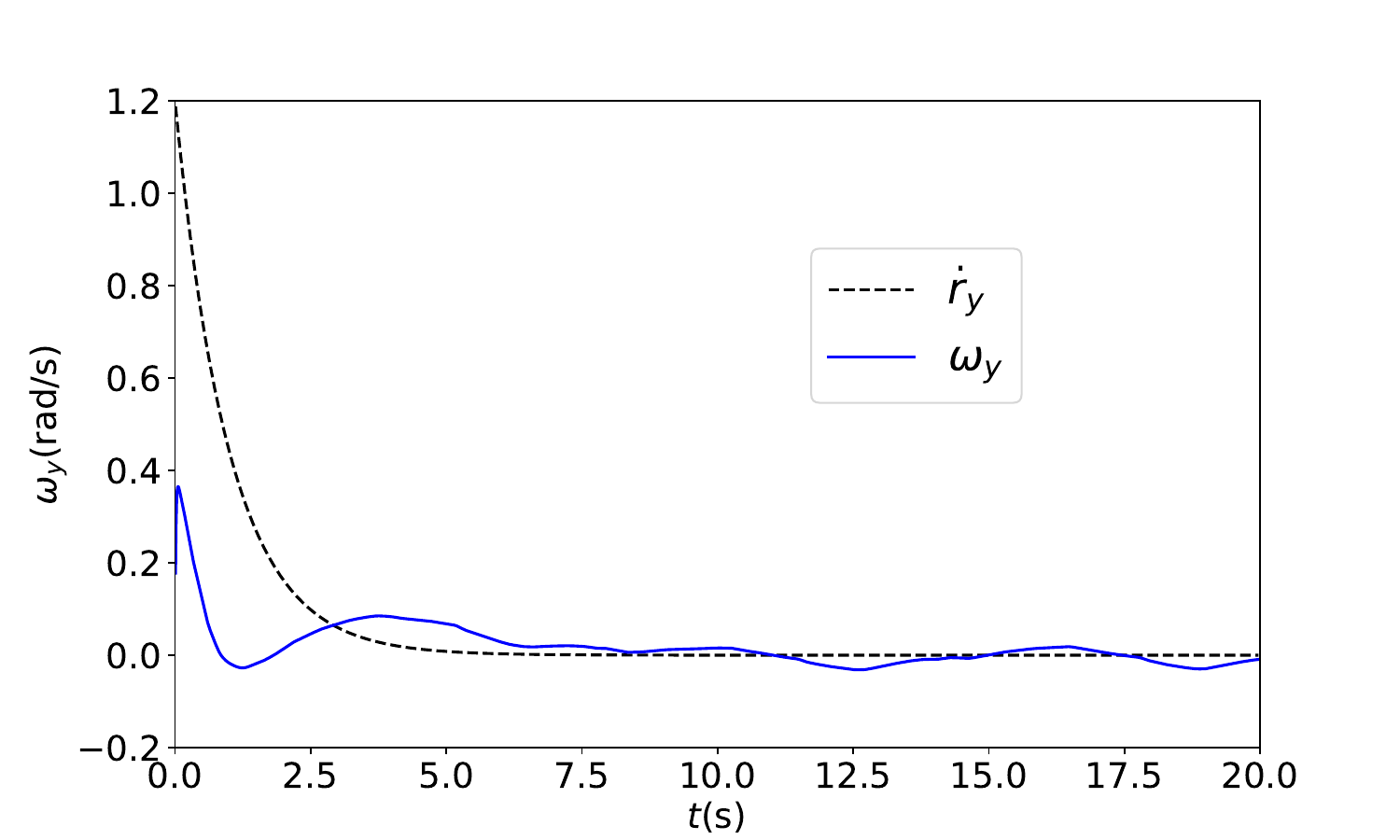}
		\caption{\label{fig:15}}
	\end{subfigure}
	\begin{subfigure}{0.4\textwidth}
		\includegraphics[scale=0.23]{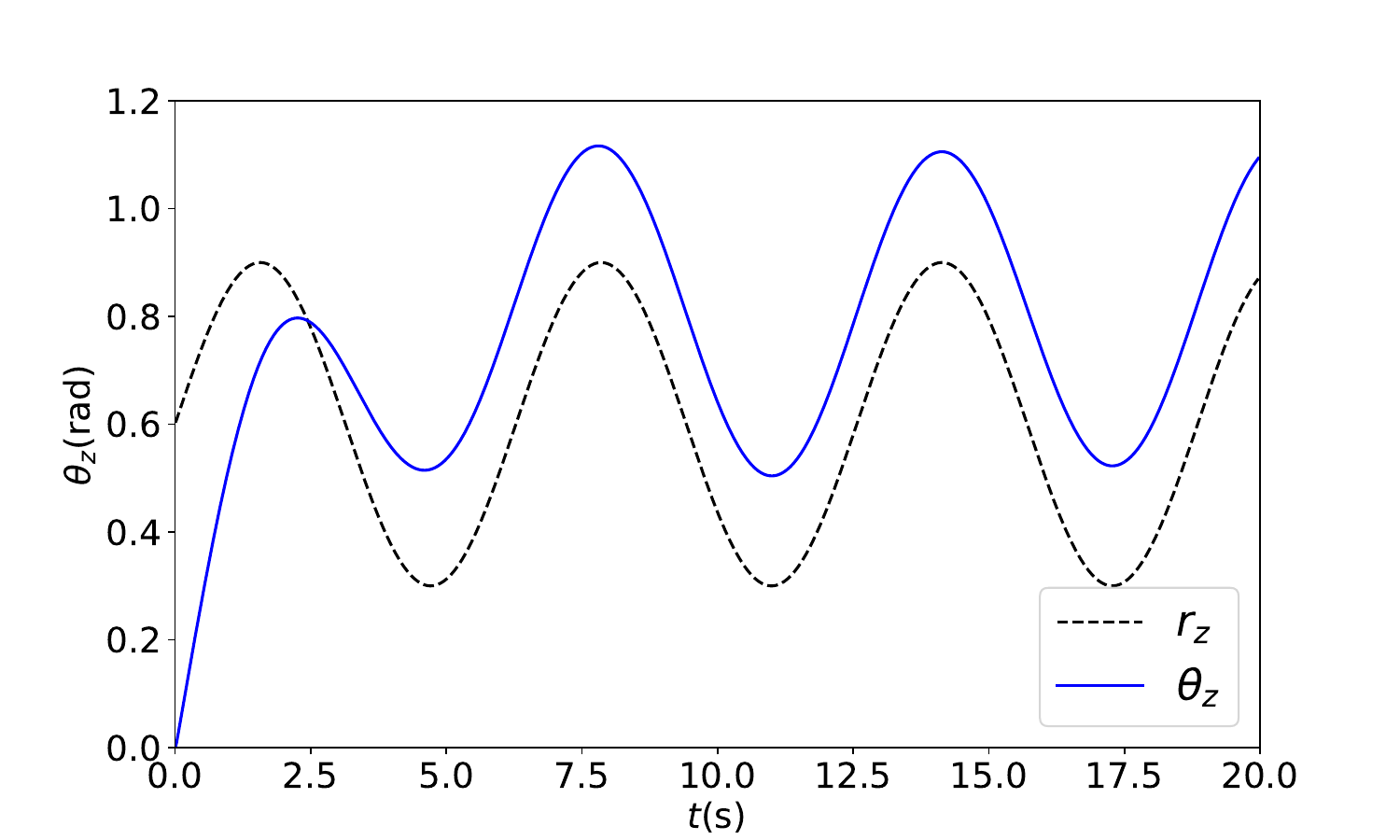}
		\caption{\label{fig:13}}
	\end{subfigure}
	\begin{subfigure}{0.4\textwidth}
		\includegraphics[scale=0.23]{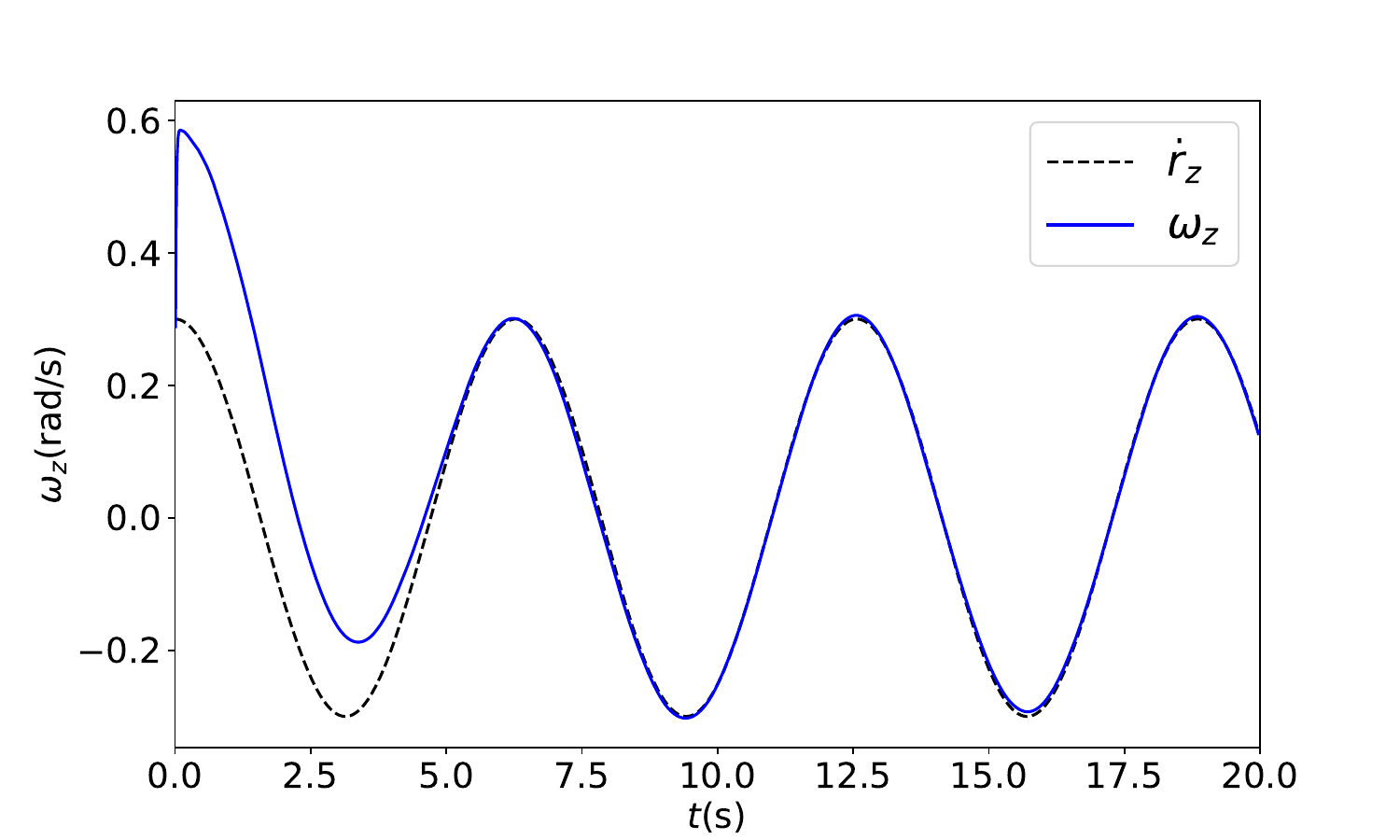}
		\caption{\label{fig:16}}
	\end{subfigure}
	\caption{Tracking property in the test environment after training in a single environment.  Time evolution of: (\subref{fig:01}) $\theta_x$, (\subref{fig:04}) $\omega_x$, (\subref{fig:02}) $\theta_y$ (\subref{fig:05}) $\omega_y$, (\subref{fig:03}) $\theta_z$, (\subref{fig:06}) $\omega_z$.}
	\label{fig:1}
\end{figure}  
\begin{figure}[h]
	\centering
	\begin{subfigure}{0.4\textwidth}
		\includegraphics[scale=0.23]{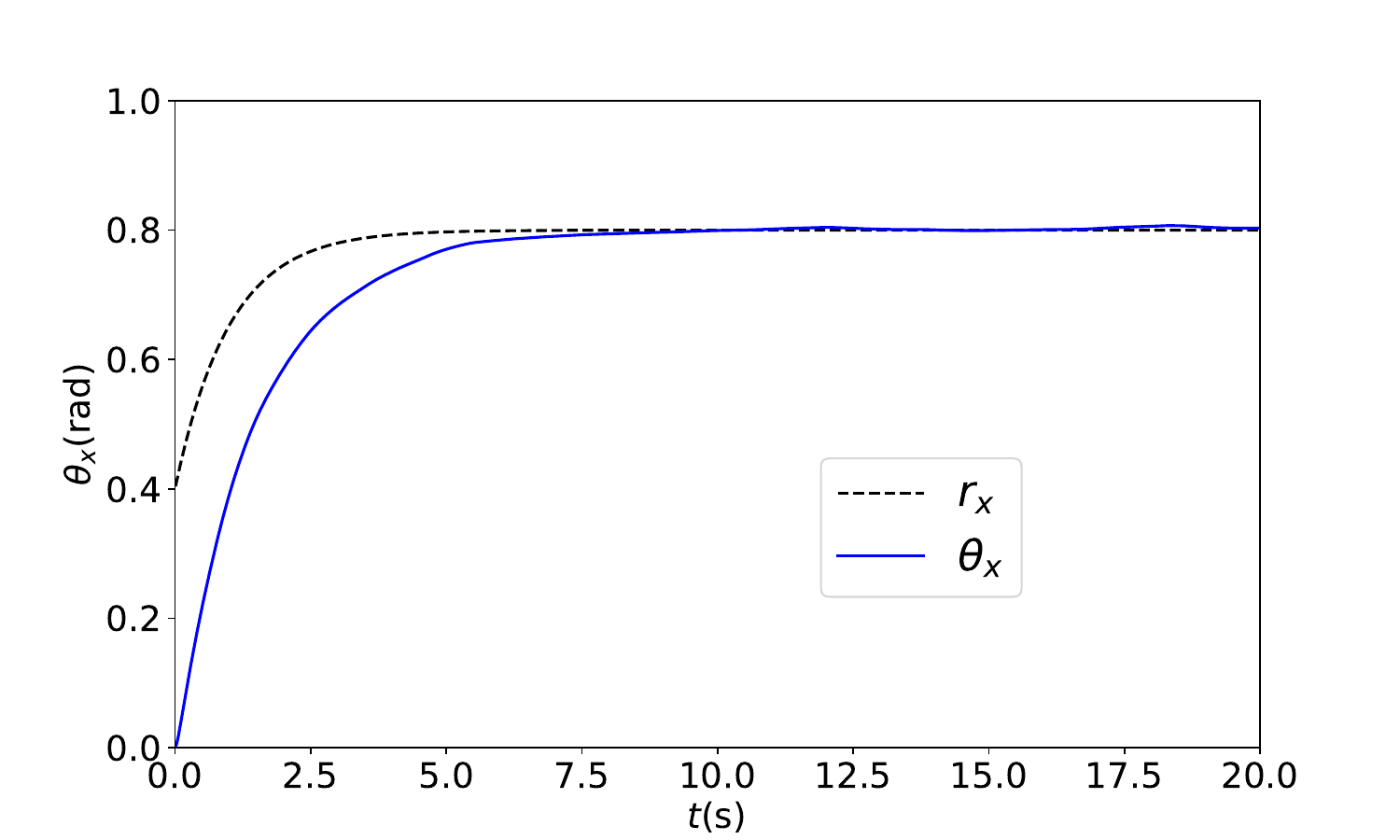}
		\caption{\label{fig:21}}
	\end{subfigure}
	\begin{subfigure}{0.4\textwidth}
		\includegraphics[scale=0.23]{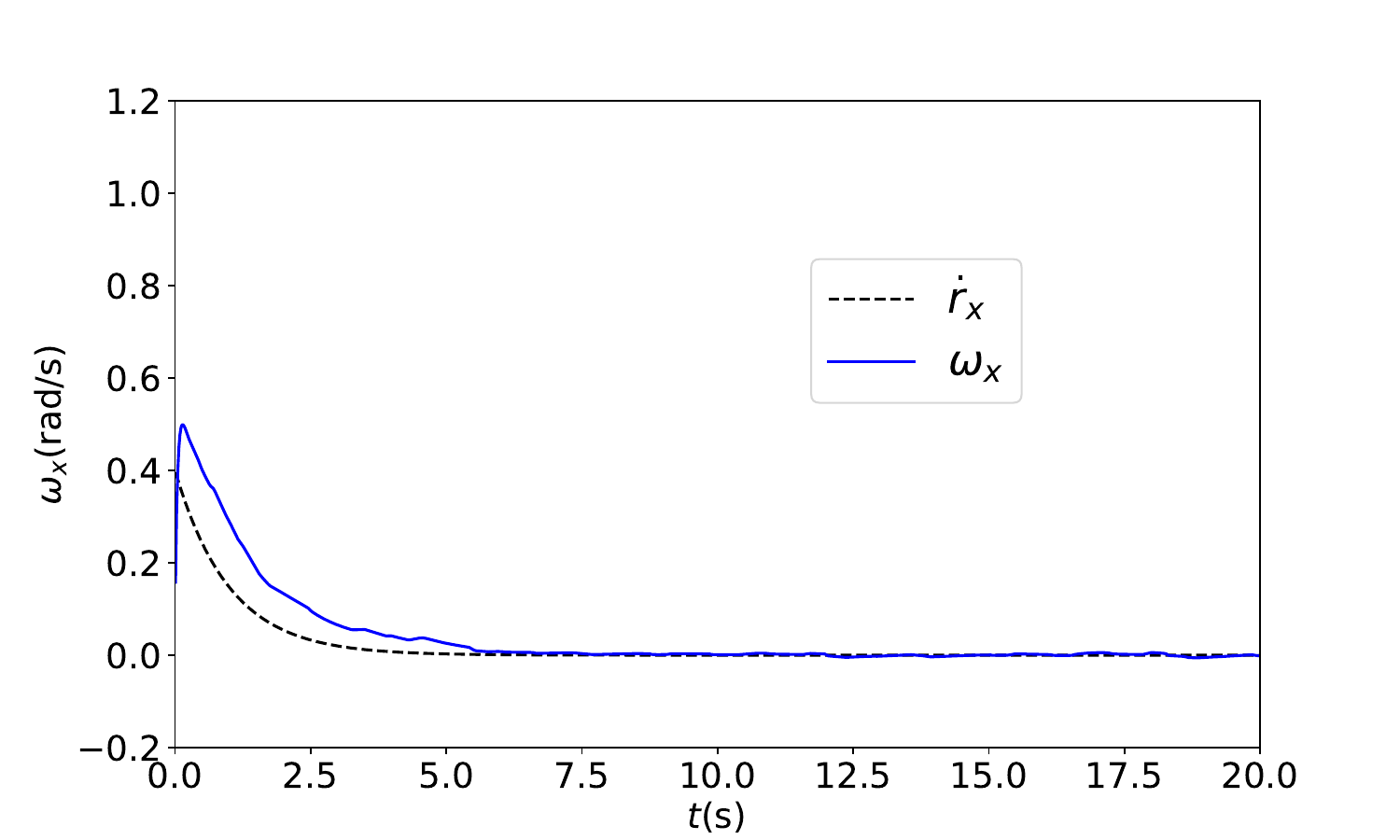}
		\caption{\label{fig:24}}
	\end{subfigure}
	\begin{subfigure}{0.4\textwidth}
		\includegraphics[scale=0.23]{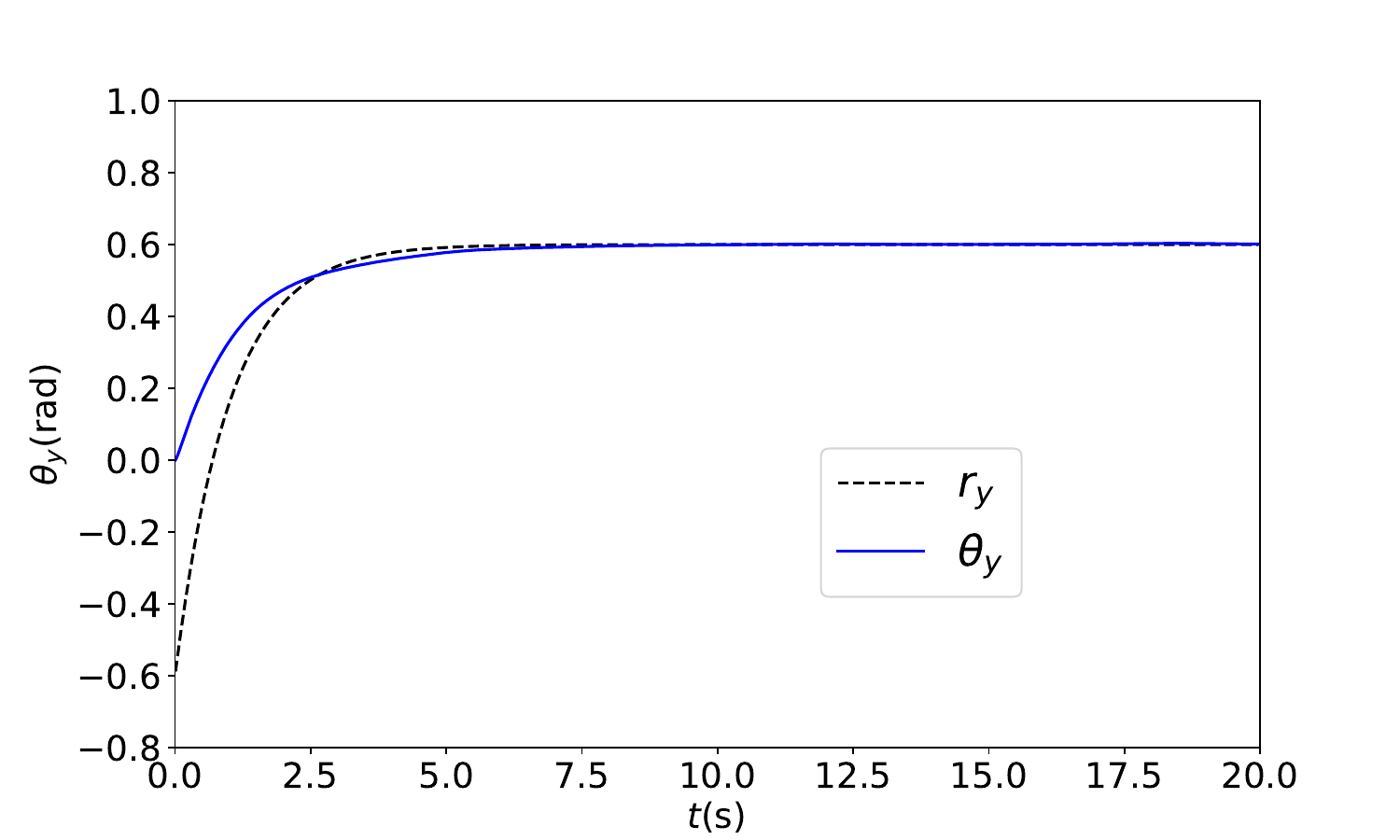}
		\caption{\label{fig:22}}
	\end{subfigure}
	\begin{subfigure}{0.4\textwidth}
		\includegraphics[scale=0.23]{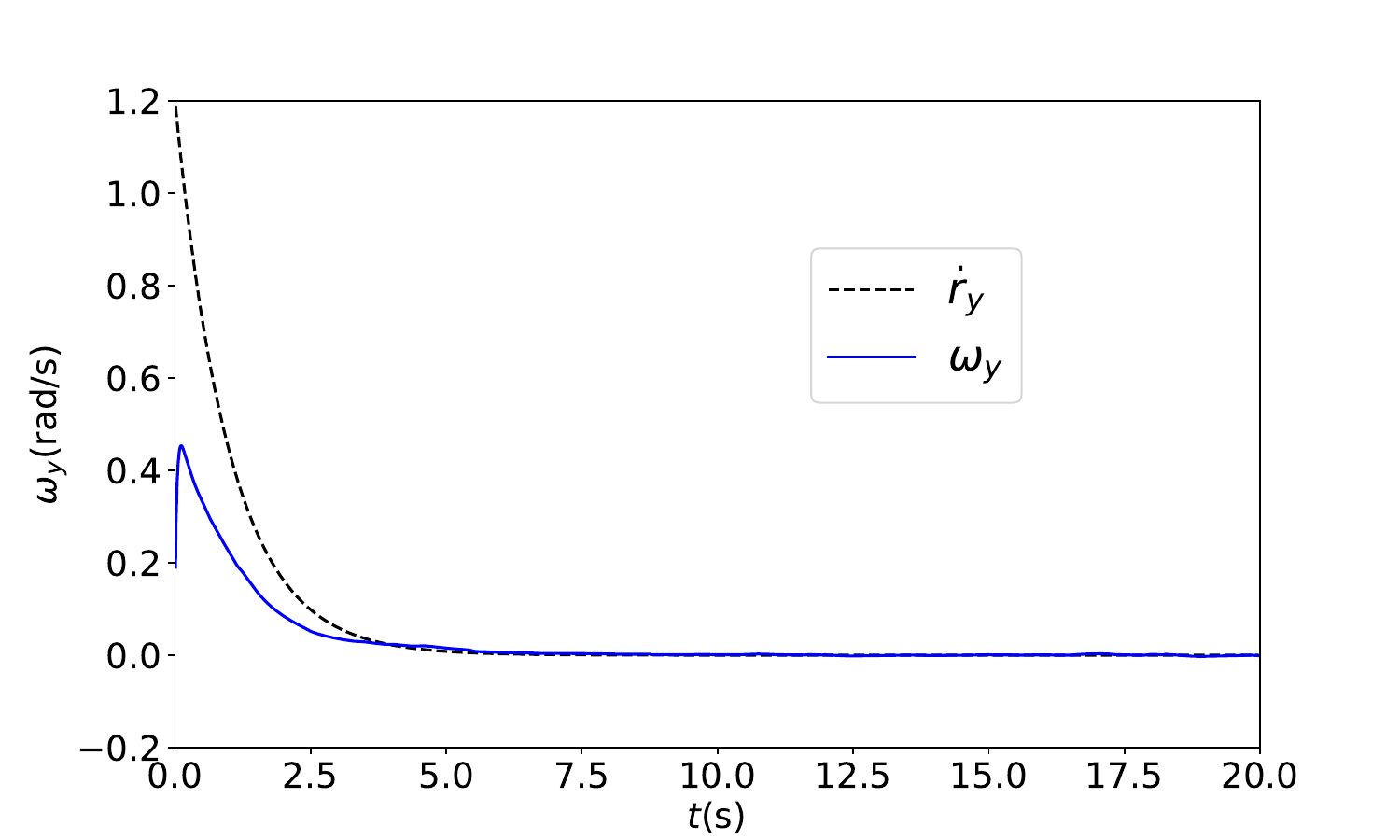}
		\caption{\label{fig:25}}
	\end{subfigure}
	\begin{subfigure}{0.4\textwidth}
		\includegraphics[scale=0.23]{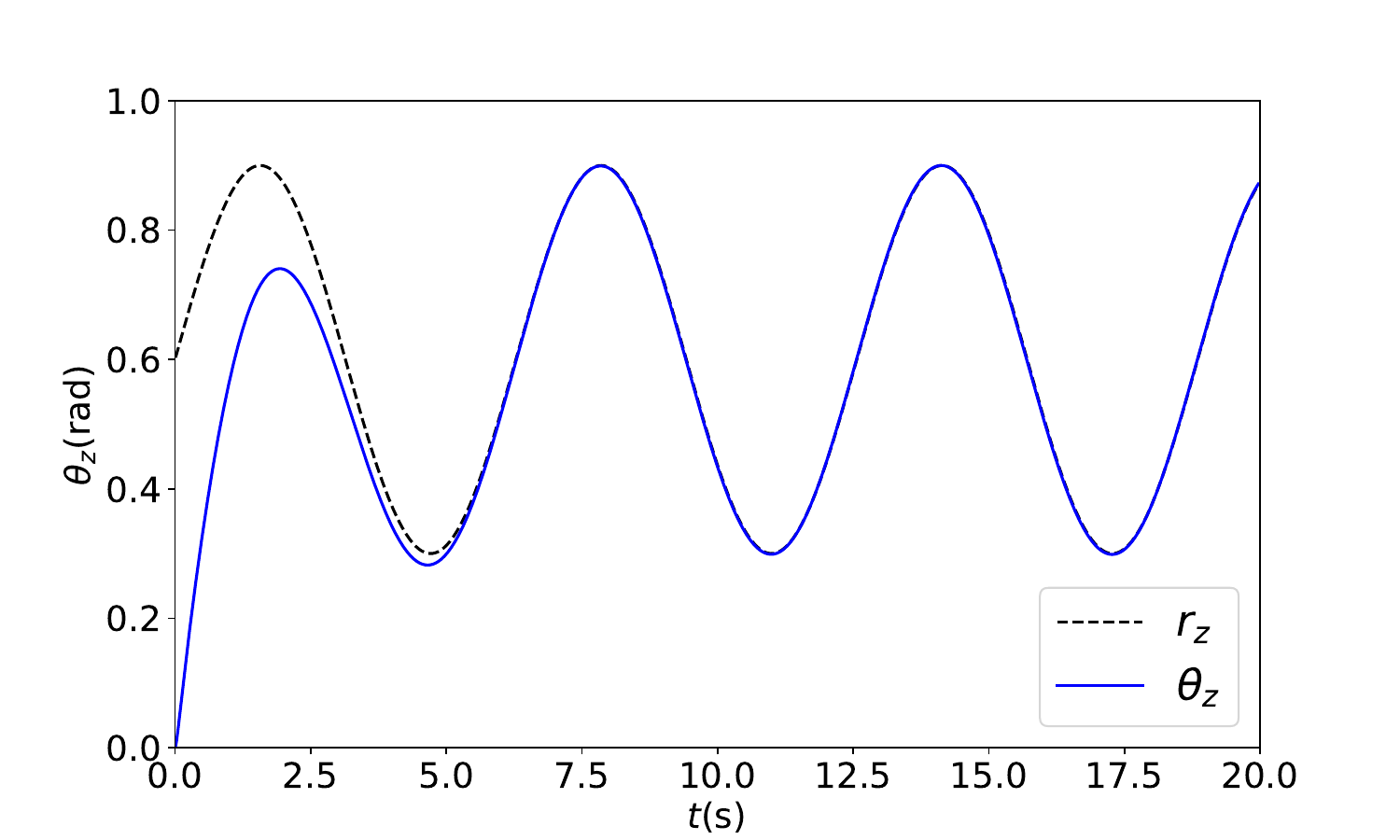}
		\caption{\label{fig:23}}
	\end{subfigure}
	\begin{subfigure}{0.4\textwidth}
		\includegraphics[scale=0.23]{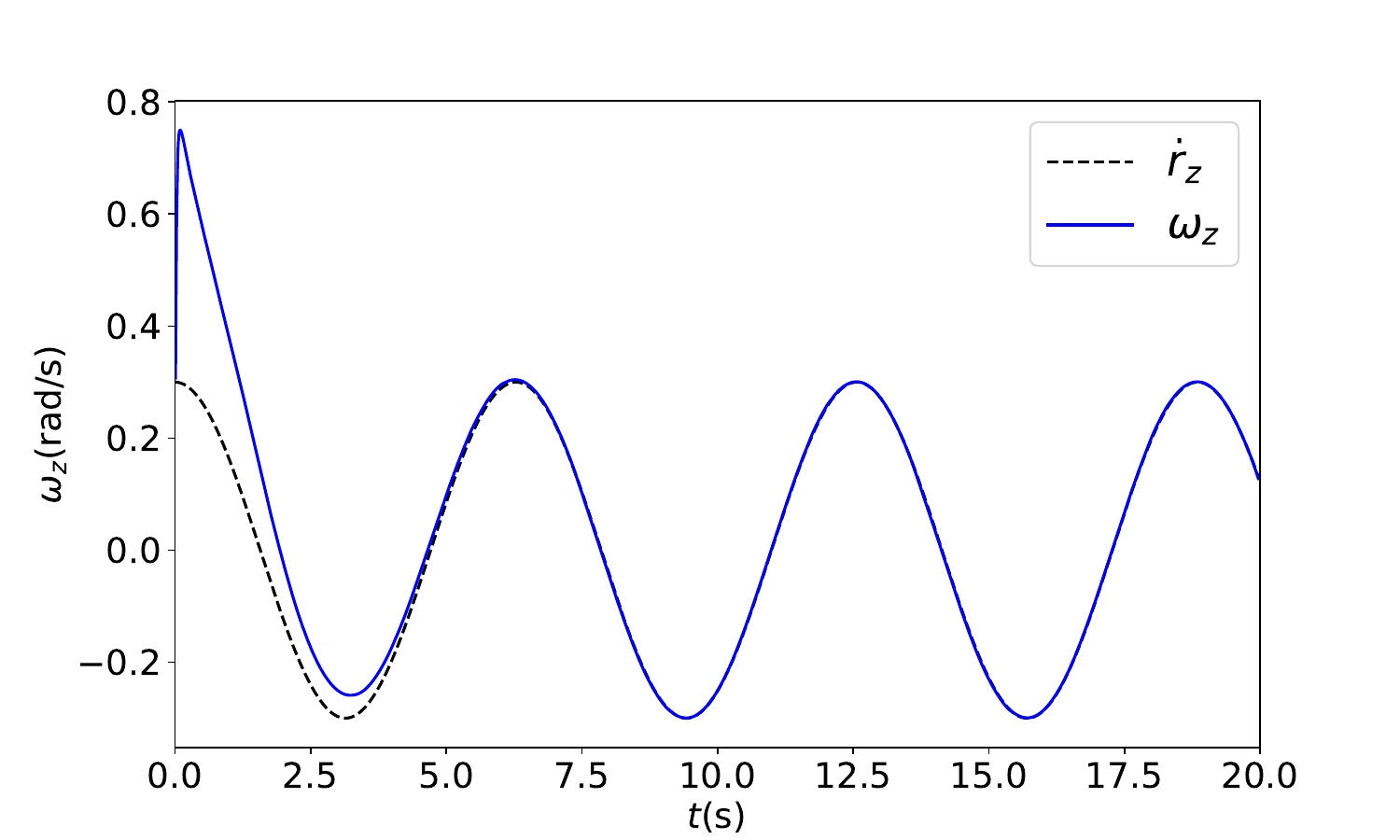}
		\caption{\label{fig:26}}
	\end{subfigure}
	\caption{Tracking property in the test environment after training in two environments. Time evolution of: (\subref{fig:01}) $\theta_x$, (\subref{fig:04}) $\omega_x$, (\subref{fig:02}) $\theta_y$ (\subref{fig:05}) $\omega_y$, (\subref{fig:03}) $\theta_z$, (\subref{fig:06}) $\omega_z$.}
	\label{fig:2}
\end{figure}

{For the model trained in a single environment, the attitude tracking error is $[0.3791,0.2320,0.2245]$ (rad). For the model trained in two environments, the error is $[0.0071,0.0033,0.0031]$ (rad).} The results show that, while training in a single environment can not ensure the tracking property, training in two different environments largely eliminates the interference from the uncontrollable environmental latents in the test environment. Compared with the results for fully controllable systems, small tracking errors still exist. This may come from the finite approximation capability of the model we use in practice and finite training data. 
\section{Discussion}
 
 In this work, we proposed an asymptotic tracking controller for a class of latent dynamic systems. This controller is based on a learned latent dynamic model with identifiable guarantees which ensure the relationship between the learned latents and the actual variables. Then, a manually designed feedback linearization controller ensures the tracking property of the closed-loop system. All the results hold for both fully controllable systems and the case where uncontrollable environmental latents exist. Experiments on a latent spacecraft attitude dynamic model have shown the effectiveness of the proposed controller. The proposed control method and corresponding theoretical results have the potential to provide safety guarantees for intelligent spacecraft.

However, there are still some main limitations in this work. First, we only consider a class of affine nonlinear systems instead of general nonlinear systems. It is a key assumption both for identifiable learning and control. How to identify and control more general nonlinear latent dynamic systems is a promising future research direction. Another main limitation is that this work is mainly a theoretical work, with synthetic datasets using an attitude dynamic model and a randomly initialized neural network as examples to verify the theoretical results. Future work can consider more challenging real-world experiments using images and videos.

\bibliographystyle{IEEEtran}
\bibliography{./sample.bib} 
\end{document}